\documentclass[%
reprint,
nofootinbib,
amsmath,amssymb,
aps,
amsthm,
pra,
]{revtex4-1}

\usepackage{ifthen}
\newboolean{arxiv}
\setboolean{arxiv}{true}
\ifthenelse{\boolean{arxiv}}{\pdfoutput=1}{}

\usepackage{theorem} 
\usepackage{ascmac}
\usepackage{bbm}
\usepackage{bm} 
\usepackage{braket}
\usepackage{dcolumn} 
\usepackage{enumerate}
\usepackage{enumitem}
\usepackage{framed}
\usepackage{graphicx} 
\ifthenelse{\boolean{arxiv}}{\usepackage{hyperref}}{\usepackage[dvipdfmx]{hyperref}}
\usepackage{longtable}
\usepackage{mathtools}
\usepackage{multirow}
\ifthenelse{\boolean{arxiv}}{}{\usepackage{pxjahyper}}
\usepackage{txfonts} 
\usepackage{xparse}

\usepackage{color}
\usepackage[most]{tcolorbox}

\ifthenelse{\boolean{arxiv}}{}{\mathtoolsset{showonlyrefs}}

\newcommand{\hypercolor}{blue}
\hypersetup{
  colorlinks,
  citecolor=\hypercolor,
  linkcolor=\hypercolor,
  urlcolor=\hypercolor,
  bookmarksopen,
  bookmarksopenlevel=4,
  bookmarksnumbered
}

\newboolean{figure}
\setboolean{figure}{true}
\newcommand{\InsertPDF}[2]{\iffigure\includegraphics[scale=#1]{#2}\fi}

\usepackage{bm}

\makeatletter
\newcounter{report}
\numberwithin{equation}{report}

\@addtoreset{figure}{report}
\makeatother

\theorembodyfont{\upshape}
\newtheorem{thm}{Theorem}[report]

\newtheorem{proposition}[thm]{Proposition}

\newtheorem{postulateno}{Postulate}

\newtheorem{lemma}[thm]{Lemma}
\newtheorem{cor}[thm]{Corollary}

\newcounter{proof}

\ExplSyntaxOn
\NewDocumentEnvironment{proof}{o}
 {
  \par\medskip
  \noindent
  \textbf{Proof~}
 }
 {\QED\par\smallskip}
\ExplSyntaxOff

\newcounter{postulate}
\renewcommand{\thepostulate}{\arabic{postulate}}
\ExplSyntaxOn
\NewDocumentEnvironment{postulate}{oo}
 {
  \refstepcounter{postulate}
  \begin{postulateno}
  \textbf{\hspace{-0.5em}\IfNoValueTF{#2}{\thepostulate}{#2} ~(\IfNoValueTF{#1}{}{#1})}
 }
 {
  \end{postulateno}
 }
\ExplSyntaxOff

\newcommand{\hA}{\hat{A}}
\newcommand{\hB}{\hat{B}}

\newcommand{\he}{\hat{e}}
\newcommand{\hf}{\hat{f}}
\newcommand{\hg}{\hat{g}}

\newcommand{\hPi}{\hat{\Pi}}
\newcommand{\hPhi}{\hat{\Phi}}
\newcommand{\hPsi}{\hat{\Psi}}

\newcommand{\hLambda}{\hat{\Lambda}}

\newcommand{\hrho}{\hat{\rho}}

\newcommand{\hsigma}{\hat{\sigma}}

\newcommand{\mC}{\mathcal{C}}
\newcommand{\mD}{\mathcal{D}}
\newcommand{\mE}{\mathcal{E}}
\newcommand{\mF}{\mathcal{F}}
\newcommand{\mG}{\mathcal{G}}

\newcommand{\mJ}{\mathcal{J}}
\newcommand{\mK}{\mathcal{K}}
\newcommand{\mL}{\mathcal{L}}

\newcommand{\mN}{\mathcal{N}}

\newcommand{\mP}{\mathcal{P}}

\newcommand{\mS}{\mathcal{S}}

\newcommand{\mU}{\mathcal{U}}

\newcommand{\mX}{\mathcal{X}}

\newcommand{\mZ}{\mathcal{Z}}

\newcommand{\ident}{\hat{1}}

\newcommand{\Real}{\mathbb{R}}
\newcommand{\Complex}{\mathbb{C}}

\newcommand{\QED}{\hspace*{0pt}\hfill $\blacksquare$}

\DeclareMathOperator{\argmax}{argmax}

\newcommand{\T}{\mathsf{T}}
\DeclareMathOperator{\Tr}{Tr}
\newcommand{\Trp}[1]{\mathop{\mathrm{Tr}_{#1}}}

\DeclareMathAlphabet{\mymathbb}{U}{BOONDOX-ds}{m}{n}
\newcommand{\zero}{\mymathbb{0}}
\newcommand{\kket}[1]{|#1\rangle\!\rangle}
\newcommand{\bbra}[1]{\langle\!\langle#1|}

\newcommand{\titlename}{Quantum process discrimination with restricted strategies}

\newcommand{\opt}{\star}
\newcommand{\hphi}{\hat{\phi}}
\newcommand{\hmE}{\hat{\mE}}
\newcommand{\tmE}{\tilde{\mE}}
\newcommand{\hcJ}{\hat{\mathcal{J}}}
\newcommand{\summ}{\sum_{m=1}^M}
\newcommand{\summt}{\sum_{m=1}^3}
\renewcommand{\ident}{\mathbbm{1}}
\newcommand{\I}{I}
\newcommand{\Pos}{\mathsf{Pos}}
\newcommand{\Meas}{\mathsf{Meas}}
\newcommand{\Sep}{\mathsf{Sep}}
\newcommand{\Chn}{\mathsf{Chn}}
\newcommand{\Den}{\mathsf{Den}}
\newcommand{\DenP}{\Den^\mathsf{P}}
\newcommand{\Tester}{\mathsf{Test}}
\newcommand{\TesterM}{\Tester_3}
\newcommand{\Comb}{\mathsf{Comb}}
\newcommand{\Her}{\mathsf{Her}}
\newcommand{\Uni}{\mathsf{Uni}}
\renewcommand{\ol}{\overline}
\DeclareMathOperator{\co}{\mathsf{co}}
\DeclareMathOperator{\clco}{\ol{\co}}
\DeclareMathOperator{\coni}{\mathsf{coni}}
\DeclareMathOperator{\clconi}{\ol{\coni}}
\newcommand{\ot}{\otimes}
\newcommand{\cross}{\times}
\newcommand{\V}{V}
\newcommand{\Vt}{{\V_t}}
\newcommand{\Wt}{{\W_t}}
\newcommand{\W}{W}
\newcommand{\WA}{{W_\mathrm{A}}}

\newcommand{\tchi}{\tilde{\chi}}
\newcommand{\tU}{\tilde{U}}
\newcommand{\tV}{{\tilde{V}}}
\newcommand{\Endash}{\text{\textendash}}
\renewcommand{\c}{\circ}
\newcommand{\termdef}{\textit}
\renewcommand{\P}{\mP}
\newcommand{\Pseq}{\P_\mathrm{seq}}
\newcommand{\PG}{\P_\G}

\newcommand{\Free}{\mF}
\newcommand{\tmF}{\tilde{\Free}}
\newcommand{\G}{\mathrm{G}}
\newcommand{\mCG}{\mC_\G}
\renewcommand{\S}{\mS}
\newcommand{\SG}{{\S_\G}}

\newcommand{\inter}{\mathrm{int}}
\newcommand{\mDC}{{\mD_\mC}}
\newcommand{\mDCG}{{\mD_{\mCG}}}
\newcommand{\Ad}{\mathrm{Ad}}

\newcommand{\inv}[1]{\bar{#1}}

\newcommand{\tvarphi}{\tilde{\varphi}}
\newcommand{\gdis}{\raisebox{-.1em}{\includegraphics[scale=0.5]{text_discard.pdf}}}
\let\ast\relax
\DeclareMathOperator{\ast}{\circledast}
\setlist[enumerate]{label=\arabic*), leftmargin=3em, itemsep=0pt, parsep=0pt, labelwidth=5em}

\let\protect\relax
{\catcode`\|=\active
  \xdef\Craket{\protect\expandafter\noexpand\csname Craket \endcsname}
  \expandafter\gdef\csname Craket \endcsname#1{\begingroup
     \ifx\SavedDoubleVert\relax
       \let\SavedDoubleVert\|\let\|\BraDoubleVert
     \fi
     \mathcode`\|32768\let|\BraVert
     \left({#1}\right)\endgroup}
}

\definecolor{memo}{RGB}{128,0,255}
\definecolor{gray}{RGB}{128,128,128}
\definecolor{purple}{RGB}{192,0,192}
\definecolor{other}{RGB}{32,192,0}

\newcommand{\Discard}[1]{}

\newcommand{\EN}[1]{}

\begin{document}

\preprint{APS/123-QED}

\title{\titlename}

\affiliation{%
 Quantum Information Science Research Center, Quantum ICT Research Institute, Tamagawa University,
 Machida, Tokyo 194-8610, Japan
}%

\author{Kenji Nakahira}
\affiliation{%
 Quantum Information Science Research Center, Quantum ICT Research Institute, Tamagawa University,
 Machida, Tokyo 194-8610, Japan
}%

\date{\today}

\begin{abstract}
 The discrimination of quantum processes, including quantum states, channels,
 and superchannels, is a fundamental topic in quantum information theory.
 It is often of interest to analyze the optimal performance that can be achieved
 when discrimination strategies are restricted to a given subset of all strategies
 allowed by quantum mechanics.
 In this paper,
 we present a general formulation of the task of finding the maximum success probability
 for discriminating quantum processes as a convex optimization problem
 whose Lagrange dual problem exhibits zero duality gap.
 The proposed formulation can be applied to any restricted strategy.
 We also derive necessary and sufficient conditions for an optimal restricted strategy
 to be optimal within the set of all strategies.
 We provide a simple example in which the dual problem given by our formulation can be much easier
 to solve than the original problem.
 We also show that the optimal performance of each restricted process discrimination problem
 can be written in terms of a certain robustness measure.
 This finding has the potential to provide a deeper insight into
 the discrimination performance of various restricted strategies.
\end{abstract}

\maketitle



\section{Introduction}

Quantum processes are fundamental building blocks of quantum information theory.
The tasks of discriminating between quantum processes are of crucial importance
in quantum communication, quantum metrology, quantum cryptography, etc.
In many situations, it is reasonable to assume that the available discrimination strategies
(also known as quantum testers) are restricted to a certain subset of all possible testers
in quantum mechanics.
For example, in practical situations, we are usually concerned only with discrimination strategies
that are readily implementable with current technology.
Another example is a setting where discrimination is performed by two or more parties
whose communication is limited.
In such settings, one may naturally ask
how the performance of an optimal restricted tester can be evaluated.
To answer this question, different individual problems of distinguishing quantum states
\cite{Wal-Sho-Har-Ved-2000,Tou-Ada-Ste-2007,Chi-Dua-Hsi-2014,Ban-Cos-Joh-Rus-2015,Nak-Usu-2016-LOCC},
measurements \cite{Sed-Zim-2014,Puc-Paw-Kra-Kuk-2018,Kra-Paw-Puc-2020,Dat-Bis-Sah-Aug-2020},
and channels \cite{Aci-2001,Sac-2005,Dua-Fen-Yin-2007,Li-Qiu-2008,Har-Has-Leu-Wat-2010,
Mat-Pia-Wat-2010,Jen-Pla-2016,Dua-Guo-Li-Li-2016,Li-Zhe-Sit-Qiu-2017,Kat-Wil-2020}
have been investigated.

It is known that if all quantum testers are allowed,
then the problem of finding the maximum success probability of guessing which process was applied
can be formalized as a semidefinite programming problem,
and its Lagrange dual problem has zero duality gap \cite{Chi-2012}.
Many discrimination problems of quantum states, measurements, and channels
have been addressed through the analysis of their dual problems
\cite{Hol-1973,Bel-1975,Jez-Reh-Fiu-2002,Eld-Meg-Ver-2003,Eld-Meg-Ver-2004,Bar-Cro-2009,
Dec-Ter-2010,Ass-Car-Pie-2010,Ha-Kwo-2013,Bae-Hwa-2013,Sin-Kim-Gho-2019,Chi-2012,
Jen-Pla-2016,Puc-Paw-Kra-Kuk-2018,Kat-Wil-2020}.
However, in a general case where the allowed testers are restricted,
the problem cannot be formalized as a semidefinite programming problem.

In this paper, we provide a general method to analyze quantum process discrimination
problems in which discrimination testers are restricted to given types of testers.
We show that the task of finding the maximum success probability
for discriminating any quantum processes can be formulated as a convex optimization problem
even if the allowed testers are restricted to any subset of all testers
and that its Lagrange dual problem has zero duality gap.
It should be mentioned that, to our knowledge, a convex programming formulation applicable to any restricted strategy
has not yet been reported even in quantum state discrimination problems.
In some scenarios, the dual problem can be much easier to solve analytically or numerically
than the original problem, as we will demonstrate through a simple example.
Our approach can deal with process discrimination problems in both cases with and without
the restriction of testers within a common framework,
which makes it easy to compare their optimal values.
Note that we use the quantum mechanical notation for convenience, but
since our method essentially relies only on convex analysis,
our techniques are applicable to a general operational probabilistic theory
(including a theory that does not obey the no-restriction hypothesis \cite{Jan-Lal-2013}).

The robustness of a resource, which is a topic closely related to
discrimination problems, has been recently extensively investigated.
It is known that the robustness of a process can be seen as a measure of its advantage
over all resource-free processes in some discrimination task
\cite{Pia-Wat-2015,Nap-Bro-Cia-Pia-2016,Pia-Cia-Bro-Nap-2016,Ans-Hsi-Jai-2018,Tak-Reg-Bu-Liu-2019,
Uol-Kra-Sha-Yu-2019,Tak-Reg-2019,Osz-Bis-2019}.
Conversely, we show that the optimal performance of 
any restricted process discrimination problem
is characterized by a certain robustness measure.

\section{Preliminaries}

\subsection{Notation}

Let $N_\V$ be the dimension of a system $\V$.
$\zero$ stands for a zero matrix.
Let $\Complex$ and $\Real_+$ be, respectively, the sets of all complex and nonnegative real
numbers.
Also, let $\Her_\V$, $\Pos_\V$, $\Den_\V$, $\DenP_\V$, and $\Meas_\V$ be, respectively, the sets
of all Hermitian matrices, positive semidefinite matrices, states (i.e., density matrices), pure states,
and measurements of a system $\V$.
Let $\I_\V$ and $\ident_\V$ be, respectively, the identity matrix on $\V$
and the identity map on $\Her_\V$.
We call a quantum operation, which corresponds to a completely positive map,
a single-step process.
Let $\Pos(\V,\W)$ and $\Chn(\V,\W)$ denote, respectively, the sets of all single-step processes
and channels (i.e., completely positive trace-preserving maps) from a system $\V$ to a system $\W$.
In this paper, a one-dimensional system is identified with $\Complex$.
Also, $\Pos(\Complex,\V)$ and $\Pos(\Complex,\Complex)$ are identified with $\Pos_\V$ and
$\Real_+$, respectively.
$H_1 \ge H_2$ with Hermitian matrices $H_1$ and $H_2$ denotes that $H_1 - H_2$ is
positive semidefinite.
Given a set $\mX$ in a real Hilbert space, we denote its interior by $\inter(\mX)$,
its closure by $\ol{\mX}$, its convex hull by $\co \mX$,
its (convex) conical hull by $\coni \mX$,
and its dual cone by $\mX^*$.
$\ol{\co \mX}$ and $\ol{\coni \mX}$ are, respectively, denoted by $\clco \mX$ and $\clconi \mX$.
$x^\T$ denotes the transpose of a matrix $x$.
$\Uni_\V$ denotes the set of all unitary matrices on a system $\V$.
For a unitary matrix $U \in \Uni_\V$, let $\Ad_U$ be the unitary channel defined as
$\Ad_U(\rho) = U \rho U^\dagger$ $~(\rho \in \Pos_\V)$.
Let $\tV \coloneqq \W_T \ot \V_T \ot \cdots \ot \W_1 \ot \V_1$,
where $T$ is some positive integer.

\subsection{Processes, combs, and testers}

In this paper, we often use diagrammatic representations to provide an intuitive understanding.
A single-step process $\hf \in \Pos(\V,\W)$ is depicted by
\begin{alignat}{1}
 \InsertPDF{1.0}{single_step_process.pdf} ~\raisebox{.1em}{.}
\end{alignat}
The system $\Complex$ is represented by `no wire'.
For example, $\hrho \in \Pos_\V$ and $\he \in \Pos(\V,\Complex)$
are diagrammatically represented as
\begin{alignat}{1}
 \InsertPDF{1.0}{state_effect.pdf} ~\raisebox{.1em}{.}
\end{alignat}
Single-step processes can be linked sequentially or in parallel.
The sequential concatenation of $\hf_1 \in \Pos(\V_1,\V_2)$ and $\hf_2 \in \Pos(\V_2,\V_3)$
is a single-step process in $\Pos(\V_1,\V_3)$, denoted as $\hf_2 \circ \hf_1$.
Also, the parallel concatenation of $\hg_1 \in \Pos(\V_1,\W_1)$ and $\hg_2 \in \Pos(\V_2,\W_2)$
is a single-step process in $\Pos(\V_1 \ot \V_2,\W_1 \ot \W_2)$, denoted as $\hg_1 \ot \hg_2$.
In diagrammatic terms, they are depicted as
\begin{alignat}{1}
 \InsertPDF{1.0}{connection.pdf} ~\raisebox{.5em}{.}
 \label{eq:process_product}
\end{alignat}

We refer to a concatenation of one or more single-step processes as a quantum process.
A process represented by a concatenation of $T$ channels is referred to as
a quantum comb with $T$ time steps \cite{Chi-Dar-Per-2008}.
States, channels, and superchannels,
which are processes that transform quantum channels to quantum channels, are special cases
of quantum combs.
The concatenation of two single-step processes
$\hf_1 \in \Pos(\V_1,\W'_1 \ot \W_1)$ and $\hf_2 \in \Pos(\W'_1 \ot \V_2,\W_2)$,
denoted by the process $\hat{F} \coloneqq \hf_1 \ast \hf_2$
(where $\ast$ denotes the concatenation), is often depicted as
\begin{alignat}{1}
 \InsertPDF{1.0}{process2.pdf} ~\raisebox{1em}{.}
 \label{eq:process2}
\end{alignat}

For a process $\hmE$ expressed in the form
\begin{alignat}{1}
 \InsertPDF{1.0}{comb.pdf}
 \label{eq:comb_pdf}
\end{alignat}
with $\hLambda^{(t)} \in \Pos(\W'_{t-1} \ot \V_t, \W'_t \ot \W_t)$,
$\W'_0 \coloneqq \Complex$, and $\W'_T \coloneqq \Complex$,
its Choi-Jamio{\l}kowski representation,
which we denote by the same letter without the hat symbol, is defined as
\begin{alignat}{1}
 \mE &\coloneqq (\hmE \ot \ident_\tV)
 (\kket{\I_\tV}\bbra{\I_\tV}) \in \Pos_\tV,
\end{alignat}
where $\kket{\I_\tV} \coloneqq \sum_n \ket{n}\ket{n} \in \tV \ot \tV$.
A process $\hmE$ is uniquely specified by its Choi-Jamio{\l}kowski representation $\mE$.
$\Comb_{\W_T,\V_T,\dots,\W_1,\V_1}$ denotes the set of all
$\tau \in \Pos_{\W_T \ot \V_T \ot \cdots \ot \W_1 \ot \V_1}$ such that
there exists $\{ \tau^{(t)} \in \Pos_{\Wt \ot \Vt \ot \cdots \ot \W_1 \ot \V_1} \}_{t=1}^{T-1}$
satisfying
\begin{alignat}{1}
 \Trp{\Wt} \tau^{(t)} &= \I_\Vt \ot \tau^{(t-1)}, \quad \forall 1 \le t \le T,
 \label{eq:Comb}
\end{alignat}
where $\tau^{(0)} \coloneqq 1$ and $\tau^{(T)} \coloneqq \tau$.
Each element of $\Comb_{\W_T,\V_T,\dots,\W_1,\V_1}$ corresponds to
a comb expressed in the form of Eq.~\eqref{eq:comb_pdf}
with $\hLambda^{(t)} \in \Chn(\W'_{t-1} \ot \V_t, \W'_t \ot \W_t)$,
$\W_0 \coloneqq \Complex$, and $\W'_T \coloneqq \Complex$.
For simplicity, we often refer to elements of $\Comb_{\W_T,\V_T,\dots,\W_1,\V_1}$ as combs.
Note that the Choi-Jamio{\l}kowski representation, $\rho$, of a state $\hrho$ is
equal to $\hrho$ itself.

$\Comb_{\Complex,\W_T,\V_T,\dots,\W_1,\V_1,\Complex}$ is denoted by
$\Comb^*_{\W_T,\V_T,\dots,\W_1,\V_1}$, which is the set of all
$\tau \in \Pos_{\W_T \ot \V_T \ot \cdots \ot \W_1 \ot \V_1}$ such that
there exist $\tau^{(1)} \in \Den_{\V_1}$ and
$\{ \tau^{(t)} \in \Pos_{\Vt \ot \W_{t-1} \ot \V_{t-1} \ot \cdots \ot \W_1 \ot \V_1} \}_{t=2}^T$
satisfying
\begin{alignat}{1}
 \tau &= \I_{\W_T} \ot \tau^{(T)}, \nonumber \\
 \Trp{\Vt} \tau^{(t)} &= \I_{\W_{t-1}} \ot \tau^{(t-1)}, \quad \forall 2 \le t \le T.
 \label{eq:Phi_sum}
\end{alignat}
Let $\mCG \coloneqq \Pos_\tV^M$ and $\SG \coloneqq \Comb^*_{\W_T,\V_T,\dots,\W_1,\V_1}$.
Each element of $\SG$ corresponds to
a comb expressed in the form
\begin{alignat}{1}
 \InsertPDF{1.0}{comb2.pdf} ~\raisebox{1em}{,}
 \label{eq:comb2_pdf}
\end{alignat}
where $\hsigma_1,\dots,\hsigma_T$ are channels (in particular, $\hsigma_1$ is a state)
and ``$\gdis$'' denotes the trace.
An ensemble of processes $\{ \hPhi_m \}_{m=1}^M$ is referred to as a tester
if $\summ \hPhi_m$ is expressed in the form of Eq.~\eqref{eq:comb2_pdf}.
For each tester element $\hPhi_m$,
$\Phi_m$ denotes the Choi-Jamio{\l}kowski representation of the process $\hPhi_m^\dagger$
(where $^\dagger$ is the adjoint operator), i.e., 
\begin{alignat}{1}
 \Phi_m \coloneqq (\hPhi_m^\dagger \ot \ident_\tV)
 (\kket{\I_\tV}\bbra{\I_\tV}) \in \Pos_\tV.
\end{alignat}
$\{ \hPhi_m \}_{m=1}^M$ is a tester if $\{ \Phi_m \}_{m=1}^M \in \mCG$ and
$\summ \Phi_m \in \SG$ hold and vice versa.
We also refer to $\{ \Phi_m \}_{m=1}^M$ as a tester.
Let $\braket{\Phi_k,\mE_m} \coloneqq \Tr(\Phi_k \mE_m)$; then,
we have
\begin{alignat}{1}
 \braket{\sigma,\tau} &= 1, \quad \forall \tau \in \Comb_{\W_T,\V_T,\dots,\W_1,\V_1},
 ~\sigma \in \SG.
 \label{eq:comb_sigma_tau}
\end{alignat}
In our manuscript, processes corresponding to elements of $\SG$ and tester elements are
diagrammatically depicted in blue.

\section{Quantum process discrimination}

We first review quantum process discrimination problems where all possible testers are allowed.
We here address the problem of discriminating $M$ combs, $\hmE_1,\dots,\hmE_M$,
where each $\hmE_m$ is the concatenation of $T$ channels
$\hLambda^{(1)}_m,\dots,\hLambda^{(T)}_m$ with ancillary systems
(see Fig.~\ref{fig:process_discrimination}).
$\hmE_m$ is expressed by $\hmE_m = \hLambda^{(T)}_m \ast \cdots \ast \hLambda^{(1)}_m$.
In the particular case where, for each $m$,
$\hmE_m$ has no ancillary system and $\hLambda^{(1)}_m,\dots,\hLambda^{(T)}_m$ are
the same channel, denoted by $\hLambda_m$,
the problem reduces to the problem of discriminating $M$ channels
$\hLambda_1,\dots,\hLambda_M$ with $T$ uses.
For simplicity, we restrict ourselves to the case where each $\hmE_m$ is
a quantum comb with $T = 2$ time steps unless otherwise mentioned,
but our approach can be readily extended to the case where each $\hmE_m$ is
a more general quantum process.
As shown in Fig.~\ref{fig:process_discrimination}, to
discriminate between given combs,
we first prepare a bipartite system $\V_1 \ot \V'_1$ in an initial state $\hsigma_1$.
One part $\V_1$ is sent through the channel $\hLambda^{(1)}_m$, followed by a channel $\hsigma_2$.
After that, we send the system $\V_2$ through the channel $\hLambda^{(2)}_m$
and perform a measurement $\{ \hPi_k \}_{k=1}^M$ on the system $\W_2 \ot \V'_2$.
Such a collection of $\{\hsigma_1,\hsigma_2,\{\hPi_k\}_{k=1}^M\}$,
which is expressed as $\{ \hPhi_k \coloneqq \hPi_k \ast \hsigma_2 \ast \hsigma_1 \}_{k=1}^M$,
can be thought of as a tester.
Any discrimination strategy, including an entanglement-assisted strategy and an adaptive strategy,
can be represented by a tester%
\footnote{For the problem of discriminating quantum channels with multiple uses,
several discrimination strategies that make use of indefinite causal order
(e.g., \cite{Ore-Cos-Bru-2012,Col-Dar-Fac-Per-2012,Chi-2012-nosignalling,Chi-Dar-Per-2013,
Bav-Mur-Qui-2021-hierarchy}) are physically allowed;
however, this paper does not deal with such strategies.}.
Let $\PG$ be the set of all such testers $\Phi \coloneqq \{\Phi_k\}_{k=1}^M$,
which can be written as (see \cite{Chi-Dar-Per-2008} for details)
\begin{alignat}{1}
 \PG &= \left\{ \{ \Phi_m \}_{m=1}^M \in \mCG : \summ \Phi_m \in \SG \right\}.
 \label{eq:PG}
\end{alignat}
Note that $\tV \coloneqq \W_2 \ot \V_2 \ot \W_1 \ot \V_1$ and
\begin{alignat}{1}
 \SG &\coloneqq \left\{ \I_{\W_2} \ot \tau_2 : \tau_2 \in \Pos_{\V_2 \ot \W_1 \ot \V_1},
 \right. \nonumber \\
 &\qquad \left. \vphantom{\Pos_{\V_2 \ot \W_1 \ot \V_1}}
 \tau_1 \in \Den_{\V_1}, ~ \Trp{\V_2} \tau_2 = \I_{\W_1} \ot \tau_1 \right\}
 \label{eq:SG}
\end{alignat}
hold.
The probability that a tester $\Phi$ gives the outcome $k$
for the comb $\mE_m$ is given by $\braket{\Phi_k,\mE_m}$.
\begin{figure}[bt]
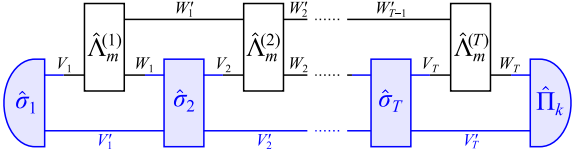

 \centering
 \InsertPDF{1.0}{process_discrimination.pdf}
 \caption{General protocol for the discrimination of quantum combs
 $\{ \hmE_m = \hLambda^{(T)}_m \ast \cdots \ast \hLambda^{(1)}_m \}_{m=1}^M$
 (plotted in black) with $T$ time steps,
 where $\W'_1,\dots,\W'_{T-1}$ are ancillary systems.
 Any physically allowed discrimination strategy can be represented by a tester (plotted in blue), which
 consists of a state $\hsigma_1$, channels $\hsigma_2,\dots,\hsigma_T$,
 and a measurement $\{\hPi_k\}_{k=1}^M$.}
 \label{fig:process_discrimination}
\end{figure}
The task of finding the maximum success probability for
discriminating the given quantum combs $\{ \mE_m \}_{m=1}^M$
with prior probabilities $\{ p_m \}_{m=1}^M$
can be formulated as an optimization problem, namely \cite{Chi-2012}
\begin{alignat}{1}
 \begin{array}{ll}
  \mbox{maximize} & \displaystyle P(\Phi) \coloneqq \summ p_m \braket{\Phi_m,\mE_m} \\
  \mbox{subject~to} & \Phi \in \PG. \\
 \end{array}
 \tag{$\mathrm{P_G}$} \label{prob:PG}
\end{alignat}

\section{Restricted discrimination}

We now consider the situation that the allowed testers are restricted to a nonempty subset
$\P$ of $\PG$; in this case, the problem is formulated as
\begin{alignat}{1}
 \begin{array}{ll}
  \mbox{maximize} & \displaystyle P(\Phi) \\
  \mbox{subject~to} & \Phi \in \P. \\
 \end{array}
 \tag{$\mathrm{P}$} \label{prob:P}
\end{alignat}
Let us interpret each tester as a vector in the real vector space $\Her_\tV^M$.
This means that one can work with linear combinations of testers $\Phi^{(1)},\Phi^{(2)},\dots$;
a tester that applies $\Phi^{(i)} \coloneqq \{ \Phi^{(i)}_k \}_{k=1}^M$
with probability $q_i$ is represented as
$\sum_i q_i \Phi^{(i)} = \{ \sum_i q_i \Phi^{(i)}_k \}_{k=1}^M$.
One can easily see that the optimal value of Problem~\eqref{prob:P}
remains the same if the feasible set $\P$ is replaced by $\clco \P$.
Indeed, an optimal solution, denoted by $\Phi^\opt \in \clco \P$, to Problem~\eqref{prob:P}
with $\P$ relaxed to $\clco \P$ can be represented as a probabilistic mixture of
$\Phi^{(1)},\Phi^{(2)},\dots \in \ol{\P}$, i.e.,
$\Phi^\opt = \sum_i \nu_i \Phi^{(i)}$ for some probability distribution $\{\nu_i\}_i$%
\footnote{Since $\P$ is bounded, $\clco \P = \co \ol{\P}$ holds.}.
Since $P(\Phi^\opt) \le P[\Phi^{(i)}]$ holds for some $i$,
$\Phi^{(i)} \in \ol{\P}$ must be an optimal solution to the relaxed problem.
Thus, Problem~\eqref{prob:P}, whose objective function is convex by construction,
is transformed into a convex optimization problem by relaxing $\P$ to $\clco \P$.
However, this relaxed problem is often very difficult to solve directly.

We find that, for any feasible set $\P$,
each tester $\Phi \in \P$ can be interpreted as
an element in some convex cone such that
the sum $\summ \Phi_m$ is in some convex set.
Specifically, we can choose a closed convex cone $\mC$
and a closed convex set $\S$ such that (see Fig.~\ref{fig:cone})
\begin{alignat}{1}
 \clco \P &= \left\{ \Phi \in \mC : \summ \Phi_m \in \S \right\},
 \quad \mC \subseteq \mCG, \quad \S \subseteq \SG.
 \label{eq:ConvP}
\end{alignat}
Such $\mC$ and $\S$ always exist
\footnote{A trivial choice is
$\mC \coloneqq \{ t \Phi : t \in \Real_+, \Phi \in \clco \P \}$ and
$\S \coloneqq \{ \summ \Phi_m : \Phi \in \clco \P \}$.}.
Equation~\eqref{eq:PG} can be regarded as a special case of this equation
with $\P = \PG$, $\mC = \mCG$, and $\S = \SG$
(note that $\clco \PG = \PG$ holds).
\begin{figure}[bt]
 \centering
 \InsertPDF{1.0}{cone.pdf}
 \caption{Schematic diagram of the closed convex hull of $\P$, which is the intersection of
 a closed convex cone $\mC$
 and a convex set $\S' \coloneqq \{ \Phi : \summ \Phi_m \in \S\}$.}
 \label{fig:cone}
\end{figure}
Let
\begin{alignat}{1}
 \mDC &\coloneqq \left\{ \chi \in \Her_\tV :
 \summ \braket{\Phi_m, \chi - p_m \mE_m} \ge 0 ~(\forall \Phi \in \mC) \right\};
\end{alignat}
then, we can easily verify that
\begin{alignat}{1}
 D_\S(\chi) &\coloneqq \max_{\varphi \in \S} \braket{\varphi,\chi}
 \ge \summ \braket{\Phi^\opt_m, \chi} \ge P(\Phi^\opt)
\end{alignat}
holds for any $\chi \in \mDC$.
The first and second inequalities follow from $\summ \Phi^\opt_m \in \S$
and $\Phi^\opt \in \mC$, respectively.
Thus, the optimal value of the following problem
\begin{alignat}{1}
 \begin{array}{ll}
  \mbox{minimize} & \displaystyle D_\S(\chi) \\
  \mbox{subject~to} & \chi \in \mDC \\
 \end{array}
 \tag{$\mathrm{D}$} \label{prob:D}
\end{alignat}
is not less than that of Problem~\eqref{prob:P}.
We refer to a feasible solution, $\chi$, to Problem~\eqref{prob:D} as proportional to some quantum comb
if $\chi$ is expressed in the form $\chi = \lambda \tchi$,
with $\lambda \in \Real_+$ and $\tchi \in \Comb_{\W_T,\V_T,\dots,\W_1,\V_1}$.
We can see that Problem~\eqref{prob:D}, which is the so-called Lagrange dual problem
of Problem~\eqref{prob:P}, has zero duality gap,
as shown in the following theorem
(proved in Appendix~\ref{append:dual}):
\begin{thm} \label{thm:dual}
 Let us arbitrarily choose a closed convex cone $\mC$ and a closed convex set $\S$
 satisfying Eq.~\eqref{eq:ConvP}; then,
 the optimal values of Problems~\eqref{prob:P} and \eqref{prob:D} are the same.
\end{thm}
%

\section{Global optimality}

\subsection{Necessary and sufficient conditions for global optimality}
\label{append:optimal_nas}

Using Theorem~\ref{thm:dual},
we can easily derive necessary and sufficient conditions for an optimal restricted strategy
to be optimal within the set of all strategies.
Given a feasible set $\P$, we now ask the question whether the optimal values of Problems~\eqref{prob:P}
and \eqref{prob:PG} coincide.
We can derive necessary and sufficient conditions for global optimality
by considering Problem~\eqref{prob:D} with $\P = \PG$
(i.e., $\mC = \mCG$ and $\S = \SG$), which is written as
\begin{alignat}{1}
 \begin{array}{ll}
  \mbox{minimize} & D_\SG(\chi) \\
  \mbox{subject~to} & \chi \in \mDCG. \\
 \end{array}
 \tag{$\mathrm{D_G}$} \label{prob:DG}
\end{alignat}
Since Theorem~\ref{thm:dual} guarantees that Problems~\eqref{prob:D} and \eqref{prob:DG},
respectively, have the same optimal values as Problems~\eqref{prob:P} and \eqref{prob:PG},
the task is to obtain necessary and sufficient conditions for
the optimal values of Problems~\eqref{prob:D} and \eqref{prob:DG} to coincide.
To this end, we have the following statement:
\begin{proposition} \label{pro:optimal2}
 Let us arbitrarily choose a closed convex cone $\mC$
 and a closed convex set $\S$ satisfying Eq.~\eqref{eq:ConvP}.
 Then, the following statements are all equivalent.
 \begin{enumerate}[label=(\arabic*)]
  \item The optimal values of Problems~\eqref{prob:P} and \eqref{prob:PG} are the same.
  \item Any optimal solution to Problem~\eqref{prob:DG} is optimal for Problem~\eqref{prob:D}.
  \item There exists an optimal solution $\chi^\opt$ to Problem~\eqref{prob:D}
        such that $\chi^\opt$ is in $\mDCG$ and is proportional to some quantum comb.
  \item There exists an optimal solution $\chi^\opt$ to Problem~\eqref{prob:D}
        such that $\chi^\opt \in \mDCG$ and $D_\S(\chi^\opt) = D_\SG(\chi^\opt)$ hold.
 \end{enumerate}
\end{proposition}
\begin{proof}
 Let $D^\opt$ and $D_\G^\opt$ be, respectively, the optimal values of
 Problems~\eqref{prob:D} and \eqref{prob:DG}
 [or, equivalently, the optimal values of Problems~\eqref{prob:P} and \eqref{prob:PG}].
 We show $(1) \Rightarrow (2)$, $(2) \Rightarrow (3)$, $(3) \Rightarrow (4)$,
 and $(4) \Rightarrow (1)$.

 $(1) \Rightarrow (2)$ :
 Let us arbitrarily choose an optimal solution $\chi^\opt$ to Problem~\eqref{prob:DG}.
 Since $\mDCG \subseteq \mDC$ holds from $\mC \subseteq \mCG$, $\chi^\opt \in \mDC$ holds.
 Also, from $\S \subseteq \SG$, we have
 $D^\opt \le D_\S(\chi^\opt) \le D_\SG(\chi^\opt)$.
 Since $D^\opt = D_\G^\opt = D_{\SG}(\chi^\opt)$ holds,
 we have $D^\opt = D_{\S}(\chi^\opt)$.
 Thus, $\chi^\opt$ is optimal for Problem~\eqref{prob:D}.

 $(2) \Rightarrow (3)$ :
 It is known that there exists an optimal solution, $\chi^\opt \in \mDCG$,
 to Problem~\eqref{prob:DG} such that $\chi^\opt$ is proportional to some
 quantum comb \cite{Chi-2012,Nak-Kat-2021-general}.
 From Statement~(2), $\chi^\opt$ is optimal for Problem~\eqref{prob:D}.

 $(3) \Rightarrow (4)$ :
 $\chi^\opt$ can be expressed as $\chi^\opt = q T$ with $q \in \Real_+$ and a quantum comb $T$.
 Since $\braket{\varphi,\chi^\opt} = q \braket{\varphi,T} = q$ holds for any $\varphi \in \SG$,
 we have $D_\S(\chi^\opt) = D_\SG(\chi^\opt)$.

 $(4) \Rightarrow (1)$ :
 We have $D^\opt = D_\S(\chi^\opt) = D_\SG(\chi^\opt) \ge D_\G^\opt$.
 Since $D^\opt \le D_\G^\opt$ holds, $D^\opt = D_\G^\opt$ must hold.
\end{proof}

\subsection{Another example of necessary and sufficient optimality conditions} \label{subsec:separable}

In some individual cases, necessary and sufficient conditions for global optimality
can also be derived from Theorem~\ref{thm:dual}.
To give an example, let us consider the problem of discriminating quantum channels
$\{ \hmE_m \}_{m=1}^M \in \Chn(\V_1,\W_1)$ with a single use
(i.e., $T = 1$) in which a state input to the channel is restricted to be separable
(see Fig.~\ref{fig:examples_sep}).
Since we can assume, without loss of generality, that the input state is
a pure state of the system $\V_1$, the optimal value $P^\opt$ of Problem~\eqref{prob:P}
is written as
\begin{alignat}{1}
 P^\opt &\coloneqq \max_{\phi \in \DenP_{\V_1}} \max_{\Pi \in \Meas_{\W_1}}
 \summ p_m \braket{\Pi_m, \hmE_m(\phi)}.
\end{alignat}
Since the dual of the discrimination problem in which an input state is fixed to $\phi$
is formulated as Problem~\eqref{prob:D} with $\mC = \mCG$ and
$\S = \left\{ \I_{\W_1} \ot \phi^\T \right\}$,
Theorem~\ref{thm:dual} gives
\begin{alignat}{1}
 \max_{\Pi \in \Meas_{\W_1}} \summ p_m \braket{\Pi_m, \hmE_m(\phi)}
 &= \min_{\chi \in \mDCG} \braket{\I_{\W_1} \ot \phi^\T, \chi}, \nonumber \\
 & \qquad \forall \phi \in \DenP_{\V_1},
\end{alignat}
and thus
\begin{alignat}{1}
 P^\opt &= \max_{\phi' \in \DenP_{\V_1}} \min_{\chi \in \mDCG} \braket{\I_{\W_1} \ot \phi', \chi}.
\end{alignat}
Also, the optimal value of Problem~\eqref{prob:DG} is expressed by
\begin{alignat}{1}
 \min_{\chi \in \mDCG} \max_{\rho \in \Den_{\V_1}} \braket{\I_{\W_1} \ot \rho^\T, \chi}
 &= \min_{\chi \in \mDCG} \max_{\phi \in \DenP_{\V_1}} \braket{\I_{\W_1} \ot \phi^\T, \chi}
 \nonumber \\
 &= \min_{\chi \in \mDCG} \max_{\phi' \in \DenP_{\V_1}} \braket{\I_{\W_1} \ot \phi', \chi}.
\end{alignat}
Thus, globally optimal discrimination is achieved
without entanglement if and only if the following max-min inequality holds as an equality:
\begin{alignat}{1}
 \max_{\phi' \in \DenP_{\V_1}} \min_{\chi \in \mDCG} \braket{\I_{\W_1} \ot \phi', \chi}
 &\le \min_{\chi \in \mDCG} \max_{\phi' \in \DenP_{\V_1}} \braket{\I_{\W_1} \ot \phi', \chi}.
 \label{eq:sep_maxmin}
\end{alignat}
\begin{figure}[bt]
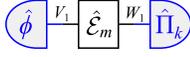

 \centering
 \InsertPDF{1.0}{examples_sep.pdf}
 \caption{Channel discrimination problems with a single use
 in which a state input to the channel is restricted to be separable.
 We can assume that the input state is a pure state of the system $\V_1$,
 i.e., a tester consists of $\hphi \in \DenP_{\V_1}$ and $\{ \hPi_m \}_{m=1}^M \in \Meas_{\W_1}$.}
 \label{fig:examples_sep}
\end{figure}

\subsection{Sufficient condition for a nonadaptive tester to be globally optimal} \label{subsec:nonadaptive}

Given a process discrimination problem that has a certain symmetry,
we present a sufficient condition for a nonadaptive tester to be globally optimal.
We here limit our discussion to a specific type of symmetries
(see \cite{Nak-Kat-2021-general} for a more general case).
Note that several related results in particular cases have been reported
\cite{Chi-2012,Jen-Pla-2016,Zhu-Pir-2020}.

Let $\mG$ be a group with the identity element $e$.
Let $\varpi \coloneqq \{ \varpi_g \}_{g \in \mG}$ be a group action of $\mG$ on $\{1,\dots,M\}$,
i.e., a set of maps on $\{1,\dots,M\}$
satisfying $\varpi_{gh}(m) = \varpi_g[\varpi_h(m)]$ and $\varpi_e(m) = m$
for any $g, h \in \mG$ and $m \in \{1,\dots,M\}$.
Given any natural number $T$, we consider a set
\begin{alignat}{1}
 \mU &\coloneqq \left\{ \mU_g \coloneqq \Ad_{U^{(T)}_g \ot \tU^{(T)}_g \ot \cdots \ot
 U^{(1)}_g \ot \tU^{(1)}_g} \right\}_{g \in \mG},
\end{alignat}
where, for each $t \in \{1,\dots,T\}$, $\mG \ni g \mapsto U^{(t)}_g \in \Uni_{\W_t}$ and
$\mG \ni g \mapsto \tU^{(t)}_g \in \Uni_{\V_t}$ are
projective unitary representations of $\mG$.
We will refer to an ensemble of $M$ combs
$\{ \mE_m \}_{m=1}^M \subset \Comb_{\W_T,\V_T,\dots,\W_1,\V_1}$
as \termdef{$(\mG,\mU,\varpi)$-covariant} if
\begin{alignat}{1}
 \mU_g(\mE_m) &= \mE_{\varpi_g(m)}, \quad \forall g \in \mG
 \label{eq:mE_cov}
\end{alignat}
holds.

We will call a tester each of whose output systems is one part of
a bipartite system in a maximally entangled pure state (see Fig.~\ref{fig:discrimination_maxent})
a \termdef{tester with maximally entangled pure states}.
Such a tester is obviously nonadaptive.
\begin{figure}[bt]
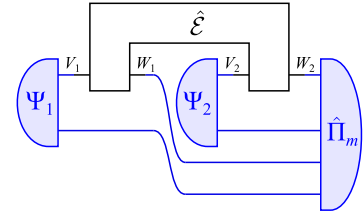

 \centering
 \InsertPDF{1.0}{discrimination_maxent.pdf}
 \caption{Tester with maximally entangled pure states (in the case of $T = 2$),
 which consists of maximally entangled pure states $\Psi_t$ and a measurement $\{\hPi_m\}$.
 We can assume, without loss of generality, that each $\Psi_t$ is a
 generalized Bell state $\kket{\I_\Vt}\bbra{\I_\Vt} / N_\Vt$.}
 \label{fig:discrimination_maxent}
\end{figure}
Let $\P$ be the set of testers with maximally entangled pure states;
then, it follows that Eq.~\eqref{eq:ConvP} holds with
\begin{alignat}{1}
 \mC &\coloneqq \mCG, \quad
 \S \coloneqq \left\{ \I_\tV / \prod_{t=1}^T N_\Vt \right\}.
 \label{eq:C_maxent}
\end{alignat}
Note that $\clco \P = \P$ holds in this case.
We obtain the following proposition.
\begin{proposition} \label{pro:optimal_maxent}
 Assume that, for each $t \in \{1,\dots,T\}$,
 there exists a group $\mG^{(t)}$ that has a projective unitary representation
 $\mG^{(t)} \ni g \mapsto U^{(t)}_g \in \Uni_\Wt$ and
 an irreducible projective unitary representation
 $\mG^{(t)} \ni g \mapsto \tU^{(t)}_g \in \Uni_\Vt$.
 Let $\mG \coloneqq \mG^{(T)} \times \mG^{(T-1)} \times \dots \times \mG^{(1)}$,
 $\mU \coloneqq \{ \Ad_{U^{(T)}_{g_T} \ot \tU^{(T)}_{g_T} \ot \dots \ot
 U^{(1)}_{g_1} \ot \tU^{(1)}_{g_1}} \}_{(g_T,\dots,g_1) \in \mG}$,
 and $\varpi$ be some group action on $\mG$ on $\{1,\dots,M\}$.
 If $\{ \mE_m \}_{m=1}^M$ is $(\mG,\mU,\varpi)$-covariant,
 then there exists a globally optimal tester with maximally entangled pure states.
\end{proposition}
See Ref.~\cite{Nak-Kat-2021-general} for some examples and for more general results.

\section{Example} \label{sec:example}

In several problems, Theorem~\ref{thm:dual} provides an efficient way to find the optimal value
of Problem~\eqref{prob:P}.
Note that the difficulty of solving Problem~\eqref{prob:D} depends on the choice of $\mC$ and $\S$.
In this section, we illustrate the usefulness of Theorem~\ref{thm:dual} in the following simple example.

Let us consider the problem of discriminating three qubit channels
$\hLambda_1,\hLambda_2,\hLambda_3$ with $T = 2$ uses,
in which case $\V_1$, $\W_1$, $\V_2$, and $\W_2$ are all qubit systems
and $\hmE_m = \hLambda_m \ast \hLambda_m$ (i.e., $\mE_m = \Lambda_m \ot \Lambda_m$) holds.
Assume that the prior probabilities are equal and that
each $\hLambda_m$ is the unitary channel represented by
$\hLambda_m(\rho) = U^m \rho U^{-m}$,
where $U \coloneqq \mathrm{diag}(1,\omega)$ and $\omega \coloneqq \exp(2\pi \sqrt{-1} / 3)$.
Then, we have
\begin{alignat}{1}
 \Lambda_m &=
 \begin{bmatrix}
  1 & 0 & 0 & \omega^{-m} \\
  0 & 0 & 0 & 0 \\
  0 & 0 & 0 & 0 \\
  \omega^m & 0 & 0 & 1 \\
 \end{bmatrix}.
\end{alignat}
We consider the case where a tester is restricted to a sequential one.
The set of all sequential testers in $\PG$, denoted by $\Pseq$, is
(see Fig.~\ref{fig:example_seq_qubit} and Appendix~\ref{append:seq_Pseq})
\begin{alignat}{1}
 \Pseq &\coloneqq \left\{ \left\{ \sum_j B^{(j)}_m \ot A_j \right\}_{m=1}^3 :
 \{A_j\} \in \Tester, ~\{ B^{(j)}_m \}_m \in \TesterM \right\},
\end{alignat}
where
\begin{alignat}{1}
 \Tester_M &\coloneqq \left\{ \{B_m\}_{m=1}^M \subset \Pos_4 : \summ B_m =
 \begin{bmatrix} 1 & 0 \\ 0 & 1 \\ \end{bmatrix} \ot \rho,
 ~ \rho \in \Den_2 \right\},
\end{alignat}
$\Tester \coloneqq \bigcup_{M=1}^\infty \Tester_M$,
and $\Pos_n$ and $\Den_n$ are, respectively, the sets of all positive semidefinite
and density matrices of order $n$.
\begin{figure}[bt]
 \centering
 \InsertPDF{1.0}{example_seq_qubit.pdf}
 \caption{Discrimination scheme for $\{ \hmE_m = \hLambda_m \ast \hLambda_m \}_{m=1}^3$
 with a sequential tester.
 Consider the task to be performed by two parties, Alice and Bob.
 Alice prepares a quantum state $\hrho_\mathrm{A} \in \Den_{\V_1 \ot \V'_1}$,
 feeds the system $\V_1$ into $\hLambda_m$,
 and performs a measurement $\{ \hPsi_j \}_{j \in \mJ}$ on $\W_1 \ot \V'_1$,
 where $\mJ$ is a set, which may contain any number of elements.
 According to its outcome, $j$, Bob prepares a state
 $\hrho_\mathrm{B}^{(j)} \in \Den_{\V_2 \ot \V'_2}$,
 feeds the system $\V_2$ into $\hLambda_m$,
 and performs a measurement $\{ \hPi^{(j)}_k \}_{k=1}^3$ on $\W_2 \ot \V'_2$.
 $\{ \hA_j \coloneqq \hPsi_j \ast \hrho_\mathrm{A} \}_{j \in \mJ}$ and
 $\{ \hB^{(j)}_k \coloneqq \hPi^{(j)}_k \ast \hrho_\mathrm{B}^{(j)} \}_{k=1}^3$
 $~(\forall j \in \mJ)$ are testers.}
 \label{fig:example_seq_qubit}
\end{figure}
Problem~\eqref{prob:P} with $\P = \Pseq$ can be written as
the following non-convex programming problem:
\begin{alignat}{1}
 \begin{array}{ll}
  \mbox{maximize} & \displaystyle \frac{1}{3} \summt \sum_j \braket{B^{(j)}_m,\Lambda_m}
   \braket{A_j,\Lambda_m} \\
  \mbox{subject~to} & \{A_j\} \in \Tester, ~\{ B^{(j)}_m \}_m \in \TesterM. \\
 \end{array}
 \label{prob:P_seq}
\end{alignat}
This problem is very hard to solve due to two main reasons:
i) both $\{A_j\}$ and $\{ B^{(j)}_m \}_m$ $~(\forall j)$ need to be optimized
and ii) how many elements an optimal tester $\{ A_j \}$ has is unknown.

Here, we pay attention to the fact that Eq.~\eqref{eq:ConvP} with
\begin{alignat}{1}
 \mC &\coloneqq \left\{ \left\{ \sum_j B^{(j)}_m \ot A_j \right\}_{m=1}^3 :
 A_j \in \Pos_4, ~\{ B^{(j)}_m \}_m \in \TesterM \right\}, \nonumber \\
 \S &\coloneqq \SG
 \label{eq:C_chsequential_sm}
\end{alignat}
holds (see Appendix~\ref{append:seq_ConvP}).
In this situation, Problem~\eqref{prob:D} is expressed as
\begin{alignat}{1}
 \begin{array}{ll}
  \mbox{minimize} & D_\SG(\chi) \\
  \mbox{subject~to} & \displaystyle \summt \Braket{\sum_j B^{(j)}_m \ot A_j,
   \chi - \frac{1}{3} \Lambda_m \ot \Lambda_m} \ge 0 \\
  &\quad \left[ \forall A_j \in \Pos_4, ~\{ B^{(j)}_m \}_m \in \TesterM \right] \\
 \end{array}
 \label{prob:D_chsequential}
\end{alignat}
with $\chi \in \Her_\tV$.
After some algebra, this problem is reduced to (see Appendix~\ref{append:seq_Dex2})
\begin{alignat}{1}
 \begin{array}{ll}
  \mbox{minimize} & \lambda \\
  \mbox{subject~to} & \displaystyle
   \begin{bmatrix}
    \lambda & 0 \\
    0 & \lambda
   \end{bmatrix}
   \ge \frac{1}{3} \summt \braket{B_m,\Lambda_m}
   \begin{bmatrix}
    1 & \omega^{-m} \\
    \omega^m & 1
   \end{bmatrix} \\
  & \quad \left( \forall \{B_m\} \in \TesterM \right) \\
 \end{array}
 \label{prob:D_seq_ex2}
\end{alignat}
with $\lambda \in \Real_+$.
This problem is much easier to solve than Problem~\eqref{prob:P_seq}.
Also, Theorem~\ref{thm:dual} guarantees that the optimal value of
Problem~\eqref{prob:D_seq_ex2}, which is numerically found to be around $0.933$,
is equal to the maximum success probability.
Note that in the case where any physically allowed discrimination strategy can be used,
we can easily see that the three channels can be perfectly distinguished with two uses
(see Appendix~\ref{append:seq_perfect}).

We restrict our discussion here to the discrimination problem for symmetric unitary qubit channels,
but our method can be applied to more general combs.
Other examples of different restricted strategies are shown in Appendix~\ref{append:restrict}.

\section{Relationship with robustnesses}

In resource theory, robustness has been used as a measure of the resourcefulness of a quantum comb,
such as a state, measurement, or channel.
For a given closed set $\Free$, called a free set, and
a closed convex cone $\mK$ of $\Her_\tV$,
the robustness of a comb $\mE \in \Pos_\tV$ against $\mK$ can be
defined as \cite{Reg-2017,Chi-Gou-2019}
\begin{alignat}{1}
 R^\Free_\mK(\mE) &\coloneqq \inf \left\{ \lambda \in \Real_+ : \frac{\mE + \lambda\mE'}{1+\lambda} \in \Free,
 ~ \mE' \in \mK \right\}.
 \label{eq:RK}
\end{alignat}
$R^\Free_\mK(\mE)$ can be intuitively interpreted as the minimal amount,
$\lambda$, of mixing with a process, $\mE' \in \mK$,
necessary in order for the mixed and renormalized process,
$(\mE + \lambda \mE') / (1 + \lambda)$, to be in $\Free$.
As already mentioned in the introduction, it has been shown that the robustness of $\mE$
is characterized as a measure of the advantage of $\mE$ over all
the processes in $\Free$ in some discrimination problem
\cite{Pia-Wat-2015,Nap-Bro-Cia-Pia-2016,Pia-Cia-Bro-Nap-2016,Ans-Hsi-Jai-2018,Tak-Reg-Bu-Liu-2019,
Uol-Kra-Sha-Yu-2019,Tak-Reg-2019,Osz-Bis-2019} (see also Appendix~\ref{append:robustness}).
Note that this problem is somewhat different from
a process discrimination problem that this Letter deals with.
Conversely, we show that the optimal value of Problem~\eqref{prob:P}
is characterized by a robustness measure.

For the problem of discriminating quantum combs $\{ \hmE_m \}_{m=1}^M$
with prior probabilities $\{ p_m \}_{m=1}^M$,
let us suppose that a party, Alice, chooses a state $\ket{m}\bra{m}$
with the probability $p_m$, where $\{ \ket{m} \}$ is the standard basis of a classical system $\WA$,
and sends the corresponding comb $\hmE_m$ to another party, Bob.
The Choi-Jamio{\l}kowski representation of the comb shared by Alice and Bob is expressed as
\begin{alignat}{1}
 \mE^\mathrm{ex} &\coloneqq \summ p_m \ket{m}\bra{m} \ot \mE_m \in \Pos_{\WA \ot \tV}.
 \label{eq:mEex}
\end{alignat}
Bob tries to infer which state Alice has.
When he uses a tester $\{\Phi_m\}_m$, the success probability is written as
$\summ \braket{\ket{m}\bra{m} \ot \Phi_m, \mE^\mathrm{ex}} = P(\Phi)$.
Using Theorem~\ref{thm:dual},
we can see that the optimal value of Problem~\eqref{prob:P} is characterized by a
robustness measure.

\begin{cor} \label{cor:Robustness}
 Let
 \begin{alignat}{1}
  \mK &\coloneqq \left\{ Y \in \Her_{\WA \ot \tV} :
  \summ \braket{\ket{m}\bra{m} \ot \Phi_m,Y} \ge 0 ~(\forall \Phi \in \mC) \right\}, \nonumber \\
  \Free &\coloneqq \{ \I_\WA \ot \chi' : \chi' \in \Her_\tV, ~D_\S(\chi') = 1/M \};
 \end{alignat}
 then, the optimal value of Problem~\eqref{prob:P} is equal to
 $[1 + R^\Free_\mK(\mE^\mathrm{ex})] / M$.
\end{cor}
\begin{proof}
 From Eq.~\eqref{eq:RK}, we have
 \begin{alignat}{1}
  \frac{1+R^\Free_\mK(\mE^\mathrm{ex})}{M} &=
  \inf \left\{ \frac{1+\lambda}{M} : \frac{\mE^\mathrm{ex} + \delta}{1+\lambda} = \I_\WA \ot \chi',
  ~\chi' \in \Her_\tV, \right. \nonumber \\
  & \phantom{\inf} \qquad \left. D_\S(\chi') = \frac{1}{M}, ~\delta \in \mK \right\}.
 \end{alignat}
 Letting $\chi \coloneqq (1+\lambda)\chi'$ and using some algebra,
 the right-hand side becomes
 \begin{alignat}{1}
  & \inf \left\{ D_\S(\chi) : \chi \in \Her_\tV,
  ~\I_\WA \ot \chi - \mE^\mathrm{ex} \in \mK \right\} \nonumber \\
  &\quad = \inf \left\{ D_\S(\chi) : \chi \in \Her_\tV,
  ~\summ \ket{m}\bra{m} \ot (\chi - p_m \mE_m) \in \mK \right\} \nonumber \\
  &\quad = \inf \left\{ D_\S(\chi) : \chi \in \mDC \right\}
  = D^\opt,
 \end{alignat}
 where $D^\opt$ is the optimal value of Problem~\eqref{prob:D}.
 Thus, Theorem~\ref{thm:dual} completes the proof.
\end{proof}
%
%

If $\mE^\mathrm{ex}$ belongs to the free set $\Free$,
then $p_1 = \cdots = p_M = 1/M$ and $\mE_1 = \cdots = \mE_M$ must hold,
which implies that $\Free$ can intuitively be thought of as
a set that includes all $\mE^\mathrm{ex}$
corresponding to trivial process discrimination problems.
This robustness measure indicates how far $\mE^\mathrm{ex}$ is from $\Free$.
This interpretation has the potential to provide a deeper insight into optimal
discrimination of quantum processes with restricted testers.

\section{Conclusion}

We have presented a general approach for solving
quantum process discrimination problems with restricted testers
based on convex programming.
Our analysis indicates that a dual problem exhibiting zero duality gap is obtained
regardless of the set of all restricted testers.
Necessary and sufficient conditions for an optimal restricted tester
to be globally optimal are also derived.
We have shown that the optimal value of each process discrimination problem
can be written in terms of a robustness measure.
In comparison to previous theoretical works, our approach would allow a unified analysis
for a large class of process discrimination problems
in which the allowed testers are restricted.
A meaningful direction for subsequent work would be to apply our approach to practical fields,
such as quantum communication and quantum metrology.

We thank for O.~Hirota, T.~S.~Usuda, and K.~Kato for insightful discussions.
This work was supported by JSPS KAKENHI Grant Number JP19K03658.


\appendix

\section{Proof of Theorem~\ref{thm:dual}} \label{append:dual}

We consider the following Lagrangian associated with Problem~\eqref{prob:P}:
\begin{alignat}{1}
 L(\Phi,\varphi,\chi) &\coloneqq \summ \braket{\Phi_m,\tmE_m} + \Braket{\varphi - \summ \Phi_m,\chi} \nonumber \\
 &= \braket{\varphi,\chi} - \summ \braket{\Phi_m, \chi - \tmE_m},
 \label{eq:L}
\end{alignat}
where $\Phi \in \mC$, $\varphi \in \S$, $\chi \in \Her_\tV$, and $\tmE_m \coloneqq p_m \mE_m$.
From Eq.~\eqref{eq:L}, we have
\begin{alignat}{1}
 \inf_\chi L(\Phi,\varphi,\chi) &=
 \begin{dcases}
  \summ \braket{\Phi_m,\tmE_m}, & \varphi = \summ \Phi_m, \\
  -\infty, & \mathrm{otherwise},
 \end{dcases} \nonumber \\
 \sup_\Phi L(\Phi,\varphi,\chi) &=
 \begin{dcases}
  \braket{\varphi,\chi}, & \chi \in \mDC, \\
  \infty, & \mathrm{otherwise}.
 \end{dcases}
\end{alignat}
Thus, the left- and right-hand sides of the max-min inequality
\begin{alignat}{1}
 \sup_{\Phi,\varphi} \inf_\chi L(\Phi,\varphi,\chi) &\le \inf_\chi \sup_{\Phi,\varphi} L(\Phi,\varphi,\chi)
 \label{eq:L_maxmin}
\end{alignat}
equal the optimal values of Problems~\eqref{prob:P} and \eqref{prob:D}, respectively.

To show the strong duality, it suffices to show that there exists $\Phi^\opt \in \clco \P$
such that $P(\Phi^\opt) \ge D^\opt$,
where $D^\opt$ is the optimal value of Problem~\eqref{prob:D}.
Let us consider the following set:
\begin{alignat}{1}
 \mZ &\coloneqq \left\{ \left( \{y_m + \tmE_m - \chi\}_{m=1}^M, D_\S(\chi) - d \right) :
 (\chi,y,d) \in \mZ_0 \right\} \nonumber \\
 &\subset \Her_\tV^M \times \Real,
\end{alignat}
where $y \coloneqq \{ y_m \}_{m=1}^M$ and
\begin{alignat}{1}
 \mZ_0 &\coloneqq \left\{ \left( \chi,y,d \right)
 \in \Her_\tV \times \mC^* \times \Real : d < D^\opt \right\}.
\end{alignat}
It is easily seen that $\mZ$ is a nonempty convex set.
Arbitrarily choose $(\chi,y,d) \in \mZ_0$ such that
$y_m + \tmE_m - \chi = \zero$ $~(\forall m)$;
then, $D_\S(\chi) \ge D^\opt$ holds from $\{\chi - \tmE_m\}_m \in \mC^*$,
which yields $D_\S(\chi) - d \ge D^\opt - d > 0$.
Thus, we have $(\{\zero\},0) \not\in \mZ$.
From separating hyperplane theorem \cite{Dha-Dut-2011},
there exists $(\{\Psi_m\}_{m=1}^M,\alpha) \neq (\{\zero\},0)$ such that
\begin{alignat}{1}
 \summ \braket{\Psi_m, y_m + \tmE_m - \chi} + \alpha [D_\S(\chi) - d] &\ge 0,
 \quad \forall (\chi,y,d) \in \mZ_0.
 \label{eq:separation0}
\end{alignat}
Substituting $y_m = c y'_m$ $~(c \in \Real_+, \{y'_m\}_m \in \mC^*)$
into Eq.~\eqref{eq:separation0} and taking the limit $c \to \infty$
yields $\{\Psi_m\}_m \in \mC$.
Taking the limit $d \to -\infty$ gives $\alpha \ge 0$.
To show $\alpha > 0$, assume by contradiction $\alpha = 0$.
Substituting $\chi = c \I_\tV$ $~(c \in \Real_+)$
and taking the limit $c \to \infty$ yields $\summ \Tr \Psi_m \le 0$.
From $\{\Psi_m\}_m \in \mC \subseteq \Pos_\tV^M$, $\Psi_m = \zero$ $~(\forall m)$ holds.
This contradicts $(\{\Psi_m\}_m,\alpha) \neq (\{\zero\},0)$,
and thus $\alpha > 0$ holds.
Let $\Phi^\opt_m \coloneqq \Psi_m / \alpha$; then, Eq.~\eqref{eq:separation0} yields
\begin{alignat}{1}
 \summ \braket{\Phi^\opt_m, y_m + \tmE_m - \chi} + D_\S(\chi) - d &\ge 0,
 \quad \forall (\chi,y,d) \in \mZ_0.
 \nonumber \\
 \label{eq:separation}
\end{alignat}
By substituting $\chi = c \chi'$ $~(c \in \Real_+, \chi' \in \Her_\tV)$
into Eq.~\eqref{eq:separation} and taking the limit $c \to \infty$,
we have $D_\S(\chi') \ge \summ \braket{\Phi^\opt_m,\chi'}$ $~(\forall \chi' \in \Her_\tV)$.
This implies $\summ \Phi^\opt_m \in \S$, i.e., $\Phi^\opt \in \clco \P$.
Indeed, assume by contradiction $\summ \Phi^\opt_m \not\in \S$;
then, since $\S$ is a closed convex set, from separating hyperplane theorem,
there exists $\chi' \in \Her_\tV$ such that
$\braket{\phi,\chi'} < \braket{\summ \Phi^\opt_m,\chi'}$ $~(\forall \phi \in \S)$,
which contradicts $D_\S(\chi') \ge \summ \braket{\Phi^\opt_m,\chi'}$.
Substituting $y_m = \zero$ and $\chi = \zero$ into Eq.~\eqref{eq:separation} and
taking the limit $d \to D^\opt$ yields
$P(\Phi^\opt) = \summ \braket{\Phi^\opt_m,\tmE_m} \ge D^\opt$.
\QED

Theorem~\ref{thm:dual} can be generalized to the following corollary.
\begin{cor} \label{cor:dual}
 Given $\P$, let us arbitrarily choose a subset $\mC$ of $\mCG$ and a bounded subset $\S$ of $\SG$
 such that
 \begin{alignat}{1}
  \clco \P &= \left\{ \Phi \in \clconi \mC : \summ \Phi_m \in \clco \S \right\}.
  \label{eq:ConvP2}
 \end{alignat}
 Then, the problem
 \begin{alignat}{1}
  \begin{array}{ll}
   \mbox{minimize} & \displaystyle \sup_{\varphi \in \S} \braket{\varphi,\chi} \\
   \mbox{subject~to} & \chi \in \mDC \\
  \end{array}
 \end{alignat}
 has the same optimal value as Problem~\eqref{prob:P}.
\end{cor}
\begin{proof}
 From Theorem~\ref{thm:dual}, the following problem
 \begin{alignat}{1}
  \begin{array}{ll}
   \mbox{minimize} & \displaystyle D_{\clco \S}(\chi) \\
   \mbox{subject~to} & \chi \in \mD_{\clconi \mC} \\
  \end{array}
  \label{prob:Dex}
 \end{alignat}
 has the same optimal value as Problem~\eqref{prob:P}.
 Also, it is easily seen that
 $D_{\clco \S}(\chi) = D_{\co \ol{\S}}(\chi) = \sup_{\varphi \in \S} \braket{\varphi,\chi}$
 and $\mD_{\clconi \mC} = \mDC$ hold.
\end{proof}

\section{Supplement of the example of sequential strategies} \label{append:seq}

\subsection{Formulation of $\Pseq$} \label{append:seq_Pseq}

We show that the set of all sequential testers in $\PG$ is expressed as
\begin{alignat}{1}
 \Pseq &\coloneqq \left\{ \left\{ \sum_j B^{(j)}_m \ot A_j \right\}_{m=1}^3 :
 \{A_j\} \in \Tester, ~\{ B^{(j)}_m \}_m \in \TesterM \right\}.
\end{alignat}
From Fig.~\ref{fig:example_seq_qubit},
$\Phi \in \Pseq$ holds if and only if $\hPhi_k$ is expressed in the form
\begin{alignat}{1}
 \InsertPDF{1.0}{examples_cZ.pdf} ~\raisebox{1em}{,}
 \label{eq:examples_cZ}
\end{alignat}
where $\hrho_\mathrm{A} \in \Den_{\V_1 \ot \V'_1}$, $\{\hPsi_j\}_j \in \Meas_{\W_1 \ot \V'_1}$,
$\hrho^{(j)}_\mathrm{B} \in \Den_{\V_2 \ot \V'_2}$ $(\forall j)$,
and $\{\hPi^{(j)}_m\}_{m=1}^M \in \Meas_{\W_2 \ot \V'_2}$ $(\forall j)$.
It follows that $\{ \hA_j \coloneqq \hPsi_j \ast \hrho_\mathrm{A} \}_{j \in \mJ}$ and
$\{ \hB^{(j)}_k \coloneqq \hPi^{(j)}_k \ast \hrho_\mathrm{B}^{(j)} \}_{k=1}^3$
are testers and Eq.~\eqref{eq:examples_cZ} can be rewritten as
\begin{alignat}{1}
 \InsertPDF{1.0}{examples_c_tester.pdf} ~\raisebox{1em}{.}
 \label{eq:examples_c_tester}
\end{alignat}
This gives that $\Phi \in \Pseq$ holds if and only if $\Phi$ is expressed in the form
$\Phi = \left\{ \sum_j B^{(j)}_m \ot A_j \right\}_{m=1}^3$
with
\begin{alignat}{1}
 \{ A_j \}_{j \in \mJ} \subset \Pos_{\W_1 \ot \V_1}, \quad
 \sum_{j \in \mJ} A_j \in \Comb^*_{\W_1,\V_1}
 \label{eq:Aj}
\end{alignat}
and
\begin{alignat}{1}
 \{ B^{(j)}_m \}_{m=1}^3 \subset \Pos_{\W_2 \ot \V_2}, \quad
 \sum_{m=1}^3 B^{(j)}_m \in \Comb^*_{\W_2,\V_2}, \quad \forall j \in \mJ.
 \nonumber \\
 \label{eq:Bjm}
\end{alignat}
From Eq.~\eqref{eq:Phi_sum},
Eqs.~\eqref{eq:Aj} and \eqref{eq:Bjm} are, respectively,
equivalent to $\{ A_j \} \in \Tester$ and $\{ B^{(j)}_m \}_m \in \Tester_3$.

\subsection{Derivation of Eq.~(\ref{eq:ConvP}) with Eq.~(\ref{eq:C_chsequential_sm})} \label{append:seq_ConvP}

Let $\P'$ be the right-hand side of Eq.\eqref{eq:ConvP}, i.e.,
\begin{alignat}{1}
 \P' &\coloneqq \left\{ \Phi \in \mC : \summ \Phi_m \in \S \right\}.
 \label{eq:Pp}
\end{alignat}
Since we can easily obtain
$\clco \P = \P$ and $\P \subseteq \P'$ (i.e., $\Phi \in \mC$ and $\summ \Phi_m \in \S$ hold
for any $\Phi \in \P$), it suffices to show $\P' \subseteq \P$.

Let us consider $\P'$ with Eq.~\eqref{eq:C_chsequential_sm}.
Arbitrarily choose $\Phi' \in \P'$.
From $\Phi' \in \mC$, $\hPhi'_k$ is expressed in the form of Eq.~\eqref{eq:examples_c_tester}
with $A_j \in \Pos_{\W_1 \ot \V_1}$ and $\{ B^{(j)}_m \}_m \in \Tester_{\W_2,\V_2}$
$~(\forall j)$.
Arbitrarily choose $\hsigma \in \Chn(\V_2,\W_2)$; then, from
\begin{alignat}{1}
 \InsertPDF{1.0}{examples_c_tester_sum.pdf} ~\raisebox{1.5em}{,}
 \label{eq:examples_c_tester_sum}
\end{alignat}
we have
\begin{alignat}{1}
 \InsertPDF{1.0}{examples_c_tester_sum3.pdf} ~\raisebox{1em}{.}
 \label{eq:examples_c_tester_sum3}
\end{alignat}
Also, from $\sum_{k=1}^M \Phi'_k \in \SG$, we have
\begin{alignat}{1}
 \InsertPDF{1.0}{examples_c_tester_sum2.pdf}
 \label{eq:examples_c_tester_sum2}
\end{alignat}
with some $\hrho'_\mathrm{A} \in \Den_{\V_1}$.
Equations~\eqref{eq:examples_c_tester_sum3} and \eqref{eq:examples_c_tester_sum2}
yield that $\{ \hA_j \}$ is a tester.
Since $\{ \hA_j \}$ and $\{ \hB^{(j)}_m \}_m$ are testers,
$\hPhi'_k$ is expressed in the form of Eq.~\eqref{eq:examples_cZ},
i.e., $\P' \subseteq \P$.

\subsection{Derivation of Problem~(\ref{prob:D_seq_ex2})} \label{append:seq_Dex2}

There exists an optimal solution $\chi$ to Problem~\eqref{prob:D_chsequential}
expressed in the form $\chi = \lambda \tchi$ with $\lambda \in \Real_+$ and
$\tchi \in \Comb_{\W_2,\V_2,\W_1,\V_1}$ (see \cite{Nak-Kat-2021-general}).
Note that from Eq.~\eqref{eq:comb_sigma_tau}, $D_\S(\tchi) = 1$ holds for any
$\tchi \in \Comb_{\W_2,\V_2,\W_1,\V_1}$.
Since $\tchi$ is a comb, we can see that
\begin{alignat}{1}
 X &\coloneqq \summt \Trp{\W_2 \ot \V_2} \left[ \left[ B^{(j)}_m \ot \I_2 \right] \tchi \right]
 \in \Comb_{\W_1,\V_1}
 \label{eq:XComb}
\end{alignat}
is independent of the measurement $\{ B^{(j)}_m\}_{m=1}^3$.
Conversely, for any $X \in \Comb_{\W_1,\V_1}$, there exists $\tchi \in \Comb_{\W_2,\V_2,\W_1,\V_1}$
satisfying Eq.~\eqref{eq:XComb}.
Thus, Problem~\eqref{prob:D_chsequential} is rewritable as
\begin{alignat}{1}
 \begin{array}{ll}
  \mbox{minimize} & \lambda \\
  \mbox{subject~to} & \displaystyle \lambda X \ge \frac{1}{3} \summt \braket{B_m,\Lambda_m} \Lambda_m
   \quad \left( \forall \{B_m\} \in \TesterM \right) \\
 \end{array}
 \label{prob:D_seq_ex1}
\end{alignat}
with $\lambda \in \Real_+$ and $X \in \Comb_{\W_1,\V_1}$.
Due to the symmetry of
\begin{alignat}{1}
 \Lambda_m &=
 \begin{bmatrix}
  1 & 0 & 0 & \omega^{-m} \\
  0 & 0 & 0 & 0 \\
  0 & 0 & 0 & 0 \\
  \omega^m & 0 & 0 & 1 \\
 \end{bmatrix},
\end{alignat}
we can assume without loss of generality that
\begin{alignat}{1}
 X &\coloneqq
 \begin{bmatrix}
  1 & 0 & 0 & 0 \\
  0 & 0 & 0 & 0 \\
  0 & 0 & 0 & 0 \\
  0 & 0 & 0 & 1 \\
 \end{bmatrix}
\end{alignat} 
holds.
This reduces Problem~\eqref{prob:D_seq_ex1} to Problem~\eqref{prob:D_seq_ex2}.

\subsection{Perfect distinguishability} \label{append:seq_perfect}

We here show that $\{ \mE_m \coloneqq \Lambda_m \ot \Lambda_m \}_{m=1}^3$
can be perfectly distinguished if any physically allowed discrimination strategy can be performed.
Assume that the prior probabilities are equal; then,
the maximum success probability is equal to the optimal value of the following problem
\cite{Chi-2012}:
\begin{alignat}{1}
 \begin{array}{ll}
  \mbox{minimize} & \lambda \\
  \mbox{subject~to} & \displaystyle \lambda \tchi \ge \mE_m / 3 \\
 \end{array}
 \label{prob:D_chsequential_global}
\end{alignat}
with $\lambda \in \Real_+$ and $\tchi \in \Comb_{\W_2,\V_2,\W_1,\V_1}$.
Note that $\mE_m$ is given by
\setcounter{MaxMatrixCols}{20}
\begin{alignat}{1}
 \mE_m &=
 \begin{bmatrix}
  1 & 0 & 0 & \omega^{-m} & 0 & 0 & 0 & 0 & 0 & 0 & 0 & 0 & \omega^{-m} & 0 & 0 & \omega^m \\
  0 & 0 & 0 & 0 & 0 & 0 & 0 & 0 & 0 & 0 & 0 & 0 & 0 & 0 & 0 & 0 \\
  0 & 0 & 0 & 0 & 0 & 0 & 0 & 0 & 0 & 0 & 0 & 0 & 0 & 0 & 0 & 0 \\
  \omega^m & 0 & 0 & 1 & 0 & 0 & 0 & 0 & 0 & 0 & 0 & 0 & 1 & 0 & 0 & \omega^{-m} \\
  0 & 0 & 0 & 0 & 0 & 0 & 0 & 0 & 0 & 0 & 0 & 0 & 0 & 0 & 0 & 0 \\
  0 & 0 & 0 & 0 & 0 & 0 & 0 & 0 & 0 & 0 & 0 & 0 & 0 & 0 & 0 & 0 \\
  0 & 0 & 0 & 0 & 0 & 0 & 0 & 0 & 0 & 0 & 0 & 0 & 0 & 0 & 0 & 0 \\
  0 & 0 & 0 & 0 & 0 & 0 & 0 & 0 & 0 & 0 & 0 & 0 & 0 & 0 & 0 & 0 \\
  0 & 0 & 0 & 0 & 0 & 0 & 0 & 0 & 0 & 0 & 0 & 0 & 0 & 0 & 0 & 0 \\
  0 & 0 & 0 & 0 & 0 & 0 & 0 & 0 & 0 & 0 & 0 & 0 & 0 & 0 & 0 & 0 \\
  0 & 0 & 0 & 0 & 0 & 0 & 0 & 0 & 0 & 0 & 0 & 0 & 0 & 0 & 0 & 0 \\
  0 & 0 & 0 & 0 & 0 & 0 & 0 & 0 & 0 & 0 & 0 & 0 & 0 & 0 & 0 & 0 \\
  \omega^m & 0 & 0 & 1 & 0 & 0 & 0 & 0 & 0 & 0 & 0 & 0 & 1 & 0 & 0 & \omega^{-m} \\
  0 & 0 & 0 & 0 & 0 & 0 & 0 & 0 & 0 & 0 & 0 & 0 & 0 & 0 & 0 & 0 \\
  0 & 0 & 0 & 0 & 0 & 0 & 0 & 0 & 0 & 0 & 0 & 0 & 0 & 0 & 0 & 0 \\
  \omega^{-m} & 0 & 0 & \omega^m & 0 & 0 & 0 & 0 & 0 & 0 & 0 & 0 & \omega^m & 0 & 0 & 1 \\
 \end{bmatrix}.
\end{alignat}
Due to the symmetry of $\mE_m$,
we can assume, without loss of generality, that $\lambda \tchi$ is expressed in the form
\begin{alignat}{1}
 \lambda \tchi &=
 \begin{bmatrix}
  \lambda & 0 & 0 & 0 & 0 & 0 & 0 & 0 & 0 & 0 & 0 & 0 & 0 & 0 & 0 & 0 \\
  0 & 0 & 0 & 0 & 0 & 0 & 0 & 0 & 0 & 0 & 0 & 0 & 0 & 0 & 0 & 0 \\
  0 & 0 & 0 & 0 & 0 & 0 & 0 & 0 & 0 & 0 & 0 & 0 & 0 & 0 & 0 & 0 \\
  0 & 0 & 0 & \lambda & 0 & 0 & 0 & 0 & 0 & 0 & 0 & 0 & y & 0 & 0 & 0 \\
  0 & 0 & 0 & 0 & 0 & 0 & 0 & 0 & 0 & 0 & 0 & 0 & 0 & 0 & 0 & 0 \\
  0 & 0 & 0 & 0 & 0 & 0 & 0 & 0 & 0 & 0 & 0 & 0 & 0 & 0 & 0 & 0 \\
  0 & 0 & 0 & 0 & 0 & 0 & 0 & 0 & 0 & 0 & 0 & 0 & 0 & 0 & 0 & 0 \\
  0 & 0 & 0 & 0 & 0 & 0 & 0 & 0 & 0 & 0 & 0 & 0 & 0 & 0 & 0 & 0 \\
  0 & 0 & 0 & 0 & 0 & 0 & 0 & 0 & 0 & 0 & 0 & 0 & 0 & 0 & 0 & 0 \\
  0 & 0 & 0 & 0 & 0 & 0 & 0 & 0 & 0 & 0 & 0 & 0 & 0 & 0 & 0 & 0 \\
  0 & 0 & 0 & 0 & 0 & 0 & 0 & 0 & 0 & 0 & 0 & 0 & 0 & 0 & 0 & 0 \\
  0 & 0 & 0 & 0 & 0 & 0 & 0 & 0 & 0 & 0 & 0 & 0 & 0 & 0 & 0 & 0 \\
  0 & 0 & 0 & y & 0 & 0 & 0 & 0 & 0 & 0 & 0 & 0 & \lambda & 0 & 0 & 0 \\
  0 & 0 & 0 & 0 & 0 & 0 & 0 & 0 & 0 & 0 & 0 & 0 & 0 & 0 & 0 & 0 \\
  0 & 0 & 0 & 0 & 0 & 0 & 0 & 0 & 0 & 0 & 0 & 0 & 0 & 0 & 0 & 0 \\
  0 & 0 & 0 & 0 & 0 & 0 & 0 & 0 & 0 & 0 & 0 & 0 & 0 & 0 & 0 & \lambda \\
 \end{bmatrix}
\end{alignat}
with $y \in \Complex$.
Note that since
$\mE_m = \Ad_{U^{m-n} \ot \I_2 \ot U^{m-n} \ot \I_2}(\mE_n)$ $~(\forall m,n \in \{1,2,3\})$
holds, there exists an optimal solution $(\lambda,\tchi)$ to Problem~\eqref{prob:D_chsequential_global}
satisfying $\tchi = \Ad_{U^k \ot \I_2 \ot U^k \ot \I_2}(\tchi)$
$~(\forall k \in \{1,2,3\})$,
as will be shown in Lemma~\ref{lemma:symmetry}.
After some simple algebra, we can see that $\lambda \tchi \ge \mE_m / 3$ is equivalent to
\begin{alignat}{1}
 3 \lambda^2 - 4\lambda + 3 y\lambda - 2y &\ge 0, ~ \lambda \ge y \ge \frac{4}{3} - 2\lambda.
\end{alignat}
We obtain the minimal $\lambda$ satisfying these inequalities,
which is the optimal value of Problem~\eqref{prob:D_chsequential_global}, to be 1.
Thus, $\{ \mE_m \coloneqq \Lambda_m \ot \Lambda_m \}_{m=1}^3$
are perfectly distinguishable.

We here give another proof.
Let us consider an ensemble of three states $\{ \hrho_m \}_{m=1}^3$ expressed as
\begin{alignat}{1}
 \InsertPDF{1.0}{example_perfect.pdf}
 \label{eq:exampls_perfect}
\end{alignat}
[i.e., $\hrho_m \coloneqq (\hLambda_m \ast \hsigma_2 \ast \hLambda_m)(\hsigma_1)$]
with $\hsigma_1 \in \Den_{\V_1 \ot \V'_1}$ and $\hsigma_2 \in \Chn(\W_1 \ot \V'_1, \V_2 \ot \W'_2)$.
We choose
\begin{alignat}{1}
 \hsigma_1 &\coloneqq \frac{1}{2}
 \begin{bmatrix}
  1 & 0 & 0 & 1 \\
  0 & 0 & 0 & 0 \\
  0 & 0 & 0 & 0 \\
  1 & 0 & 0 & 1 \\
 \end{bmatrix}, \quad
 \hsigma_2 \coloneqq \Ad_{U'}, \nonumber \\
 U' &\coloneqq \frac{1}{\sqrt{3}}
 \begin{bmatrix}
  \sqrt{2} & 0 & 1 & 0 \\
  0 & \sqrt{2} & 0 & 1 \\
  -1 & 0 & \sqrt{2} & 0 \\
  0 & -1 & 0 & \sqrt{2} \\
 \end{bmatrix};
\end{alignat}
then, it is easily seen that $\{ \hrho_m \}_{m=1}^3$ are orthogonal.
Therefore, $\{ \mE_m \}_{m=1}^3$ are perfectly distinguishable.

\section{Applications of Theorem~\ref{thm:dual}} \label{append:restrict}

In addition to the example given in the main paper,
we will provide two other examples demonstrating the utility of
Theorem~\ref{thm:dual}.
In this section, we consider the case $T = 2$.

\subsection{First example}

The first example is the restriction to nonadaptive testers [see Fig.~\ref{fig:examples}(a)].
Let
\begin{alignat}{1}
 \mC &\coloneqq \mCG, \nonumber \\
 \S &\coloneqq \left\{ \cross_{\W_1,\V_2}(\I_{\W_2 \ot \W_1} \ot \rho') :
 \rho' \in \Den_{\V_2 \ot \V_1} \right\},
 \label{eq:C_nonadaptive}
\end{alignat}
where $\cross_{\W,\V} \in \Chn(\W \ot \V,\V \ot \W)$ is the channel
that swaps two systems $\W$ and $\V$, which is depicted by
\begin{alignat}{1}
 \InsertPDF{1.0}{cross.pdf} ~\raisebox{1em}{.}
 \label{eq:cross}
\end{alignat}
We will show that $\P'$ of Eq.~\eqref{eq:Pp} is equal to $\clco \P$
[i.e., Eq.~\eqref{eq:ConvP} holds] in the next paragraph.
Substituting Eq.~\eqref{eq:C_nonadaptive} into Problem~\eqref{prob:D} yields
the following dual problem
\begin{alignat}{1}
 \begin{array}{ll}
  \mbox{minimize} & \displaystyle \max_{\rho' \in \Den_{\V_2 \ot \V_1}}
   \braket{\cross_{\W_1,\V_2}(\I_{\W_2 \ot \W_1} \ot \rho'), \chi} \\
  \mbox{subject~to} & \chi \in \mDCG = \{ \chi \in \Pos_\tV : \chi \ge p_m \mE_m ~(\forall m) \}. \\
 \end{array}
 \label{prob:D_nonadaptive}
\end{alignat}
Note that although this problem is also formulated as the task of
discriminating $\{ \cross_{\V_2,\W_1} (\mE_m) \in \Comb_{\W_2 \ot \W_1,\V_2 \ot \V_1} \}_m$
with a single use,
the expression of Problem~\eqref{prob:D_nonadaptive} is useful for verifying
the global optimality of nonadaptive testers.
\begin{figure}[bt]
 \centering
 \InsertPDF{0.95}{examples.pdf}
 \caption{Examples of two types of restricted testers.
 (a) A nonadaptive tester,
 in which first a state $\hrho \in \Den_{\V_1 \ot \V_2 \ot \V'_1}$ is prepared, then
 the parts $\V_1$ and $\V_2$ are transmitted through a given process $\hmE_m$,
 and finally a measurement $\{\hPi_k\}_{k=1}^M \in \Meas_{\W_2 \ot \W_1 \ot \V'_1}$ is performed.
 (b) A tester performed by two parties, Alice and Bob.
 In this tester, Alice first prepares a state $\hrho_i \in \Den_{\V_1 \ot \V'_1}$
 with a probability $q_i$ and sends its one part $\V_1$ to a given process $\hmE_m$.
 She also sends the classical information $i$ to Bob, who performs a channel
 $\hsigma_i \in \Chn(\W_1,\V_2)$ depending on $i$.
 Alice finally performs a measurement $\{\hPi^{(i)}_k\}_{k=1}^M \in \Meas_{\W_2 \ot \V'_1}$.
 Only one-way classical communication from Alice to Bob is allowed.}
 \label{fig:examples}
\end{figure}

We here show $\P' = \clco \P$.
It is easily seen $\clco \P = \P$ and $\P \subseteq \P'$
(i.e., $\Phi \in \mC$ and $\summ \Phi_m \in \S$ hold for any $\Phi \in \P$).
Thus, it suffices to show $\P' \subseteq \P$.
From Fig.~\ref{fig:examples}(a),
$\Phi \in \P$ holds if and only if $\hPhi_k$ is expressed in the form
\begin{alignat}{1}
 \InsertPDF{1.0}{examples_aZ.pdf} ~\raisebox{1em}{,}
 \label{eq:examples_aZ}
\end{alignat}
where $\hrho \in \Den_{\V_1 \ot \V_2 \ot \V'_1}$ and
$\{\hPi_m\}_{m=1}^M \in \Meas_{\W_2 \ot \W_1 \ot \V'_1}$.
Arbitrarily choose $\Phi' \in \P'$.
One can easily verify $\S \subseteq \SG$, i.e., $\P' \subseteq \PG$.
Thus, $\hPhi'_k$ is expressed in the form
\begin{alignat}{1}
 \InsertPDF{1.0}{examples_a_tester.pdf} ~\raisebox{1em}{,}
\end{alignat}
where $\hsigma_1 \in \Den_{\V_1 \ot \V''_1}$, $\hsigma_2 \in \Chn(\W_1 \ot \V''_1,\V_2 \ot \V''_2)$,
and $\{ \hPi'_m \}_{m=1}^M \in \Meas_{\W_2 \ot \V''_2}$.
$\V''_1$ and $\V''_2$ are some ancillary systems.
Also, from $\summ \Phi'_m \in \S$,
there exists $\hrho' \in \Den_{\V_1 \ot \V_2}$ such that
\begin{alignat}{1}
 \InsertPDF{1.0}{examples_a.pdf} ~\raisebox{1em}{.}
\end{alignat}
Thus, $\hPhi'_k$ is expressed in the form of Eq.~\eqref{eq:examples_aZ},
where $\hrho$ is a purification of $\hrho'$.
Therefore, $\P' \subseteq \P$ holds.

\subsection{Second example}

The second example is described in Fig.~\ref{fig:examples}(b).
We want to find $\mC$ and $\S$ such that Problem~\eqref{prob:D} is easy to solve;
however, it is hard to find such $\mC$ and $\S$ satisfying Eq.~\eqref{eq:ConvP}.
Instead, we consider relaxing this equation to $\clco \P \subseteq \P'$.
We here choose
\begin{alignat}{1}
 \mC &\coloneqq \left\{ \left\{ \cross_{\V_2 \ot \W_1,\W_2}
 \left( \sum_i \sigma_i \ot A^{(i)}_m \right) \right\}_{m=1}^M : \right. \nonumber \\
 &\qquad \left. \vphantom{\left( \sum_i \sigma_i \ot A^{(i)}_m \right)}
 \sigma_i \in \Comb_{\V_2,\W_1}, ~A^{(i)}_m \in \Pos_{\W_2 \ot \V_1} \right\}
 \label{eq:C_Bobchannel}
\end{alignat}
and $\S \coloneqq \SG$.
This allows Problem~\eqref{prob:D} to be rewritten in this situation as
\begin{alignat}{1}
 \begin{array}{ll}
  \mbox{minimize} & D_\SG(\chi) \\
  \mbox{subject~to} & \displaystyle \Trp{\V_2 \ot \W_1} [\sigma (\chi - p_m\mE_m)] \ge \zero
   \\
  & \qquad (\forall 1 \le m \le M, ~\sigma \in \Comb_{\V_2,\W_1})
 \end{array}
 \label{prob:D_Bobchannel}
\end{alignat}
with $\chi \in \Her_\tV$.
We relatively easily obtain the optimal value, denoted by $D^\opt$, of this problem.
The optimal value of Problem~\eqref{prob:P} is upper bounded by $D^\opt$,
since $D^\opt$ coincides with the optimal value of Problem~\eqref{prob:P} where the feasible set
is relaxed from $\P$ to $\P'$.

\section{Proof of Proposition~\ref{pro:optimal_maxent}} \label{append:symmetry}

Before giving the proof, we first prove the following two lemmas.
\begin{lemma} \label{lemma:symmetry}
 Assume that $\{ \mE_m \}_{m=1}^M$ is $(\mG,\mU,\varpi)$-covariant.
 If
 \begin{alignat}{1}
  \mC &= \mCG, \nonumber \\
  \mU_g(\varphi) &\in \S, \quad \forall g \in \mG, ~\varphi \in \S
  \label{eq:sym_ConvS}
 \end{alignat}
 holds, then there exists an optimal solution, $\chi^\opt \in \Pos_\tV$,
 to Problem~\eqref{prob:D} such that
 \begin{alignat}{1}
  \mU_g(\chi^\opt) &= \chi^\opt, \quad \forall g \in \mG.
  \label{eq:sym_chi}
 \end{alignat}
\end{lemma}
\begin{proof}
 Let $\chi$ be an optimal solution to Problem~\eqref{prob:D}.
 From $\mC = \mCG$, we can easily see $\chi \in \Pos_\tV$.
 Also, let
 \begin{alignat}{1}
  \chi^\opt &\coloneqq \frac{1}{|\mG|}
  \sum_{g \in \mG} \mU_g(\chi) \in \Pos_\tV,
 \end{alignat}
 where $|\mG|$ is the order of $\mG$.
 Since $\mU_g \c \mU_{g'} = \mU_{gg'}$ holds for any $g,g' \in \mG$,
 one can easily see Eq.~\eqref{eq:sym_chi}.
 We have that from Eq.~\eqref{eq:mE_cov} and $\chi \ge \mE_m$,
 \begin{alignat}{1}
  \chi^\opt - \mE_m &= \frac{1}{|\mG|} \sum_{g \in \mG} \mU_g(\chi - \mE_{\varpi_{\inv{g}}(m)}) \ge \zero,
  \quad \forall m \in \{1,\dots,M\},
  \label{eq:sym_Xl1}
 \end{alignat}
 i.e., $\chi^\opt \in \mDC$,  where $\inv{g}$ is the inverse of $g$.
 Moreover, we have
 \begin{alignat}{1}
  D_\S(\chi^\opt) &= \max_{\varphi \in \S} \braket{\varphi,\chi^\opt}
  = \frac{1}{|\mG|} \max_{\varphi \in \S} \sum_{g \in \mG} \braket{\varphi,\mU_g(\chi)}
  \nonumber \\
  &\le \frac{1}{|\mG|} \sum_{g \in \mG} \max_{\varphi \in \S} \braket{\varphi,\mU_g(\chi)}
  = \frac{1}{|\mG|} \sum_{g \in \mG} \max_{\varphi \in \S} \braket{\mU_{\inv{g}}(\varphi),\chi}
  \nonumber \\
  &\le \frac{1}{|\mG|} \sum_{g \in \mG} D_\S(\chi) = D_\S(\chi),
 \end{alignat}
 where the last inequality follows from $\mU_{\inv{g}}(\varphi) \in \S$ for each $\varphi \in \S$
 and $g \in \mG$, which follows from the second line of Eq.~\eqref{eq:sym_ConvS}.
 Therefore, $\chi^\opt$ is optimal for Problem~\eqref{prob:D}.
\end{proof}
The proof of this lemma also shows that
if there exists an optimal solution that is proportional to some quantum comb, then
there also exists an optimal solution, $\chi^\opt$, that is proportional to some quantum comb
and that satisfies Eq.~\eqref{eq:sym_chi}.

\begin{lemma} \label{lemma:symmetry_irreducible}
 Assume that, for each $t \in \{1,\dots,T\}$,
 there exists a group $\mG^{(t)}$ that has two projective unitary representations
 $\mG^{(t)} \ni g \mapsto U^{(t)}_g \in \Uni_\Wt$ and
 $\mG^{(t)} \ni g \mapsto \tU^{(t)}_g \in \Uni_\Vt$.
 If $g \mapsto \tU^{(t)}_g$ is irreducible for each $t \in \{1,\dots,T\}$, then
 any $\chi^\opt \in \Pos_\tV$ satisfying
 \begin{alignat}{1}
  (\ident_{\W_T \ot \V_T \ot \cdots \ot \W_{t+1} \ot \V_{t+1}} \ot \Ad_{U^{(t)}_g \ot \tU^{(t)}_g}
  \ot \ident_{\W_{t-1} \ot \V_{t-1} \ot \cdots \ot \W_1 \ot \V_1})(\chi^\opt) = \chi^\opt,
  \nonumber \\
  \qquad \forall t \in \{1,\dots,T\}, ~g \in \mG^{(t)}
  \label{eq:sym_chi2}
 \end{alignat}
 is proportional to some quantum comb.
\end{lemma}
\begin{proof}
 It suffices to show that, for each $t \in \{1,\dots,T\}$,
 $\Trp{\Wt} \chi^\opt$ is expressed in the form $\Trp{\Wt} \chi^\opt = \I_\Vt \ot \chi^\opt_t$
 with $\chi^\opt_t \in \Pos_{X_t}$,
 where $X_t$ is the tensor product of all $\W_{t'} \ot \V_{t'}$
 with $t' \in \{1,\dots,t-1,t+1,\dots,T\}$.
 Indeed, in this case, one can easily verify that
 $\chi^\opt$ is proportional to some quantum comb.
 Let us fix $t \in \{1,\dots,T\}$.
 Also, let $\chi^\opt_s \coloneqq \Trp{X_t}[(\I_\Vt \ot s) \Trp{\Wt} \chi^\opt] \in \Pos_\Vt$,
 where $s \in \Pos_{X_t}$ is arbitrarily chosen;
 then, we have
 \begin{alignat}{1}
  \Tr \chi^\opt_s &= \braket{s, \chi'},
  \quad \chi' \coloneqq \Trp{\Wt \ot \Vt} \chi^\opt.
  \label{eq:TrChiOptS}
 \end{alignat}
 Equation~\eqref{eq:sym_chi2} gives
 $\Ad_{\tU^{(t)}_g}(\chi^\opt_s) = \chi^\opt_s$ $~[\forall g \in \mG^{(t)}]$.
 From Schur's lemma,
 $\chi^\opt_s$ must be proportional to $\I_\Vt$.
 Thus, from Eq.~\eqref{eq:TrChiOptS}, $\chi^\opt_s = \braket{s,\chi'} \I_\Vt / N_\Vt$ holds.
 We have that for any $s' \in \Pos_\Vt$,
 \begin{alignat}{1}
  \braket{s' \ot s, \Trp{\Wt} \chi^\opt} &= \braket{s', \chi^\opt_s}
  = \braket{s, \chi'} \braket{s', \I_\Vt / N_\Vt} \nonumber \\
  &= \braket{s' \ot s, \I_\Vt \ot \chi' / N_\Vt}.
  \label{eq:ss_chidt}
 \end{alignat}
 Since Eq.~\eqref{eq:ss_chidt} holds for any $s$ and $s'$,
 we have $\Trp{\Wt} \chi^\opt = \I_\Vt \ot \chi^\opt_t$
 with $\chi^\opt_t \coloneqq \chi' / N_\Vt \in \Pos_{X_t}$.
\end{proof}

Now, we are ready to prove Proposition~\ref{pro:optimal_maxent}. 
Let $\mC$ and $\S$ be defined as Eq.~\eqref{eq:C_maxent};
then, it is easily seen that Eq.~\eqref{eq:sym_ConvS} holds.
Let $e_t$ be the identity element of $\mG^{(t)}$.
From Lemma~\ref{lemma:symmetry},
there exists an optimal solution $\chi^\opt \in \Pos_\tV$
to Problem~\eqref{prob:D} satisfying Eq.~\eqref{eq:sym_chi}.
Thus, $\mU_{(e_T,\dots,e_{t+1},g_t,e_{t-1},\dots,e_1)}(\chi^\opt) = \chi^\opt$ holds
for any $t \in \{1,\dots,T\}$ and $g_t \in \mG^{(t)}$,
i.e., Eq.~\eqref{eq:sym_chi2} holds,
which implies from Lemma~\ref{lemma:symmetry_irreducible} that
$\chi^\opt$ is proportional to some quantum comb.
Therefore, Proposition~\ref{pro:optimal2} concludes the proof.

\section{Relationship between robustnesses and process discrimination problems} \label{append:robustness}

Let us consider the robustness of $\mE \in \Her_\tV$ defined by
\begin{alignat}{1}
 R^\Free_\mK(\mE) &\coloneqq \inf \left\{ \lambda \in \Real_+ : \frac{\mE + \lambda\mE'}{1+\lambda} \in \Free,
 ~ \mE' \in \mK \right\}, \quad \mE \in \Her_\tV,
 \label{eq:RK2t}
\end{alignat}
where $\mK$ $~(\subset \Her_\tV)$ is a proper convex cone
[or, equivalently, $\mK$ is a closed convex cone that is pointed (i.e., $\mK \cap -\mK = \{ \zero \}$)
and has nonempty interior] and
$\Free$ $~(\subset \Her_\tV)$ is a compact set.
Assume that the set $\{ \lambda Z : \lambda \in \Real, ~0 \le \lambda \le 1, ~Z \in \Free \}$
is convex.
In order for this value to be well-defined, we assume that $\Free \cap \inter(\mK)$ is not empty.
The so-called global (or generalized) robustness of a state $\rho \in \Den_\V$
with respect to $\Free \subseteq \Den_\V$, defined as \cite{Har-Nie-2003}
\begin{alignat}{1}
 R_{\Free}(\rho) &\coloneqq \min \left\{ \lambda \in \Real_+ : \frac{\rho + \lambda\rho'}{1+\lambda} \in \Free,
 ~\rho' \in \Den_\V \right\}, \nonumber \\
 \qquad \rho \in \Den_\V
 \label{eq:RF}
\end{alignat}
is equal to $R^\Free_{\Pos_\V}(\rho)$.
In other words, $R_{\Free}:\Den_\V \to \Real_+$ is the same function as
$R^\Free_{\Pos_\V}:\Her_\V \to \Real_+$, but is only defined on $\Den_\V$.
As an example of $R_{\Free}$, if $\Free$ is the set of all bipartite separable states,
then $R_{\Free}(\rho)$ can be understood as a measure of entanglement.
The robustness $R_{\Free}(\rho)$ is known to represent the maximum advantage that
$\rho$ provides in a certain subchannel discrimination problem
(e.g., \cite{Tak-Reg-Bu-Liu-2019,Uol-Kra-Sha-Yu-2019}).
Similarly, as will be seen in Proposition~\ref{pro:R_disc_ex},
$R^\Free_\mK(\mE)$ has a close relationship
with the maximum advantage that $\mE$ provides in a certain discrimination problem.

By letting $Z \coloneqq (\mE + \lambda\mE')/(1+\lambda)$,
we can rewrite $R^\Free_\mK(\mE)$ of Eq.~\eqref{eq:RK2t} as
\begin{alignat}{1}
 R^\Free_\mK(\mE) &= \min \{ \lambda \in \Real_+ : (1 + \lambda) Z - \mE \in \mK, ~Z \in \Free \}.
 \label{eq:RK2a}
\end{alignat}
Let
\begin{alignat}{1}
 \mN &\coloneqq \{ \mE \in \Her_\tV : \delta Z - \mE \not\in \mK ~(\forall \delta < 1, Z \in \Free) \};
\end{alignat}
then, it follows that
\begin{alignat}{1}
 R^\Free_\mK(\mE) &= \min \{ \lambda \in \Real : (1 + \lambda) Z - \mE \in \mK, ~Z \in \Free \},
 \quad \forall \mE \in \mN.
 \label{eq:RK2}
\end{alignat}

We first prove the following two lemmas.
\begin{lemma} \label{lemma:R_disc0}
 If $\mE \in \mN$ holds, then
 \begin{alignat}{1}
  1 + R^\Free_\mK(\mE) &= \max
  \{ \braket{\varphi,\mE} : \varphi \in \mK^*, \braket{\varphi,Z} \le 1 ~(\forall Z \in \Free) \}
  \nonumber \\
  \label{eq:R_dual}
 \end{alignat}
 holds.
\end{lemma}
\begin{proof}
 Let $\tmF \coloneqq \coni \Free$.
 $\eta:\tmF \to \Real_+$ denotes the gauge function of $\Free$, which is defined as
 \begin{alignat}{1}
  \eta(Y) &\coloneqq \min \{ \lambda \in \Real_+ : Y = \lambda Z, ~Z \in \Free \}, \quad Y \in \tmF.
  \label{eq:etaY}
 \end{alignat}
 Note that $\eta$ is a convex function.
 Equation~\eqref{eq:RK2} gives
 \begin{alignat}{1}
  1 + R^\Free_\mK(\mE) &= \min \{ \eta(Y) : Y - \mE \in \mK, ~Y \in \tmF \}.
  \label{eq:RK3}
 \end{alignat}
 Let us consider the following Lagrangian
 \begin{alignat}{1}
  L(Y,\varphi) &\coloneqq \eta(Y) - \braket{\varphi,Y - \mE}
  = \braket{\varphi,\mE} + \eta(Y) - \braket{\varphi,Y}
 \end{alignat}
 with $Y \in \tmF$ and $\varphi \in \mK^*$.
 We can easily verify
 \begin{alignat}{1}
  \sup_{\varphi \in \mK^*} L(Y,\varphi) &=
  \begin{dcases}
   \eta(Y), & Y - \mE \in \mK, \\
   \infty, & \mathrm{otherwise},
  \end{dcases} \nonumber \\
  \inf_{Y \in \tmF} L(Y,\varphi) &=
  \begin{dcases}
   \braket{\varphi,\mE}, & \braket{\varphi,Y'} \le \eta(Y') ~(\forall Y' \in \tmF), \\
   -\infty, & \mathrm{otherwise}.
  \end{dcases}
 \end{alignat}
 Thus, the max-min inequality
 \begin{alignat}{1}
  \inf_{Y \in \tmF} \sup_{\varphi \in \mK^*} L(Y,\varphi) &\ge
  \sup_{\varphi \in \mK^*} \inf_{Y \in \tmF} L(Y,\varphi)
 \end{alignat}
 yields
 \begin{alignat}{1}
  1 + R^\Free_\mK(\mE) &\ge \max
  \{ \braket{\varphi,\mE} : \varphi \in \mK^*, ~\braket{\varphi,Y} \le \eta(Y) ~(\forall Y \in \tmF) \}.
  \label{eq:R_dual2}
 \end{alignat}
 We now prove the equality of Eq.~\eqref{eq:R_dual2}.
 To this end, it suffices to show that there exists $\varphi \in \inter(\mK^*)$ such that
 $\braket{\varphi,Y} \le \eta(Y)$ $~(\forall Y \in \tmF)$;
 indeed, in this case, the equality of Eq.~\eqref{eq:R_dual2} follows from Slater's condition.
 Arbitrarily choose $\varphi' \in \inter(\mK^*)$
 and let $\gamma \coloneqq \sup_{Y \in \tmF \setminus \{\zero\}} [\braket{\varphi',Y}/\eta(Y)]$
 [note that $\eta(Y) > 0$ holds for any $Y \in \tmF \setminus\{\zero\}$].
 Since $\Free \cap \inter(\mK)$ is not empty,
 there exists $Y \in \Free \cap \inter(\mK)$ such that $\braket{\varphi',Y} > 0$, which yields $\gamma > 0$.
 Let $\varphi \coloneqq \gamma^{-1} \varphi' \in \inter(\mK^*)$; then,
 we can easily verify $\braket{\varphi,Y} \le \eta(Y)$ $~(\forall Y \in \tmF)$.

 It remains to show
 \begin{alignat}{1}
  \braket{\varphi,Z} \le 1 ~(\forall Z \in \Free)
  &\quad\Leftrightarrow\quad \braket{\varphi,Y} \le \eta(Y) ~(\forall Y \in \tmF).
 \end{alignat}
 We first prove ``$\Rightarrow$''.
 Arbitrarily choose $Y \in \tmF$; then,
 from Eq.~\eqref{eq:etaY}, there exists $Z \in \Free$ such that $Y = \eta(Y)Z$.
 Thus, $\braket{\varphi,Y} = \eta(Y)\braket{\varphi,Z} \le \eta(Y)$ holds.
 We next prove ``$\Leftarrow$''.
 Arbitrarily choose $Z \in \Free$.
 Since the case $\braket{\varphi,Z} \le 0$ is obvious, we may assume $\braket{\varphi,Z} > 0$.
 Let $Z' \coloneqq Z / \eta(Z)$; then, from $\eta(Z) \le 1$, $Z' \in \tmF$, and $\eta(Z') = 1$,
 we have $\braket{\varphi,Z} \le \braket{\varphi,Z'} \le \eta(Z') = 1$.
\end{proof}
\begin{lemma} \label{lemma:R_disc}
 If $\mE \in \mN$ holds, then we have
 \begin{alignat}{1}
  \max_{\varphi \in \mX \setminus \{\zero\}}
  \frac{\braket{\varphi,\mE}}{\max_{Z \in \Free} \braket{\varphi,Z}}
  &= 1 + R^\Free_\mK(\mE),
  \label{eq:R_disc}
 \end{alignat}
 where $\mX$ is any set such that the cone generated by $\mX$ is $\mK^*$, i.e.,
 $\{ \lambda \varphi : \lambda \in \Real_+, ~\varphi \in \mX\} = \mK^*$.
\end{lemma}
\begin{proof}
 Let
 \begin{alignat}{1}
  \varphi^\opt &\coloneqq \mathop{\argmax}\limits_{\varphi \in \mK^*, \Gamma(\varphi) \le 1}
  \braket{\varphi,\mE},
  \quad \Gamma(\varphi) \coloneqq \max_{Z \in \Free} \braket{\varphi,Z};
 \end{alignat}
 then, from Eq.~\eqref{eq:R_dual}, we have $\braket{\varphi^\opt,\mE} = 1 + R^\Free_\mK(\mE)$.
 Since $\Free \cap \inter(\mK)$ is not empty,
 $\Gamma(\varphi) > 0$ holds for any $\varphi \in \mK^* \setminus \{\zero\}$.
 It follows that $\Gamma(\varphi^\opt) = 1$ must hold
 [otherwise, $\tvarphi \coloneqq \varphi^\opt / \Gamma(\varphi^\opt)$
 satisfies $\braket{\tvarphi,\mE} > \braket{\varphi^\opt,\mE}$,
 $\tvarphi \in \mK^*$, and $\Gamma(\tvarphi) = 1$,
 which contradicts the definition of $\varphi^\opt$].
 For any $\varphi \in \mX \setminus \{\zero\}$,
 $\varphi' \coloneqq \varphi / \Gamma(\varphi)$ satisfies
 $\braket{\varphi,\mE} / \Gamma(\varphi) = \braket{\varphi',\mE}$
 and $\Gamma(\varphi') = 1$.
 Thus, we have
 \begin{alignat}{1}
  \max_{\varphi \in \mX \setminus \{\zero\}} \frac{\braket{\varphi,\mE}}{\Gamma(\varphi)}
  &= \max_{\substack{\varphi' \in \mX \setminus \{\zero\}, \\ \Gamma(\varphi') = 1}}
  \braket{\varphi',\mE} = \braket{\varphi^\opt,\mE} = 1 + R^\Free_\mK(\mE).
 \end{alignat}
\end{proof}

We should note that, in practical situations, many physically interesting processes belong to $\mN$.
As an example, if $\Free$ is a subset of all combs in $\Her_\tV$ and
$\inter(\mK^*) \cap \Comb^*_{\W_T,\V_T,\dots,\W_1,\V_1}$ is not empty,
then any comb in $\Her_\tV$ belongs to $\mN$.
[Indeed, arbitrarily choose
$\phi \in \inter(\mK^*) \cap \Comb^*_{\W_T,\V_T,\dots,\W_1,\V_1}$
and a comb $\mE \in \Her_\tV$;
then, from $\phi \in \Comb^*_{\W_T,\V_T,\dots,\W_1,\V_1}$,
$\braket{\phi,Z} = \braket{\phi,\mE} = 1$ $~(\forall Z \in \Free)$ holds.
Thus, for any $\delta < 1$, $\braket{\phi,\delta Z - \mE} = \delta - 1 < 0$ holds,
which yields $\delta Z - \mE \not\in \mK$.
Therefore, we have $\mE \in \mN$.]
For instance, if $\Free \subseteq \Den_\V$ and
$\Tr x > 0$ $~(\forall x \in \mK \setminus \{\zero\})$ hold,
then any $\rho \in \Den_\V$ is in $\mN$.
As another example, if $R^\Free_\mK(\mE) > 0$ holds, then $\mE \in \mN$ always holds.
[Indeed, by contraposition, assume $\mE \not\in \mN$; then,
there exists $\delta < 1$ and $Z \in \Free$ such that $\delta Z - \mE \in \mK$.
Let $\lambda^\opt \coloneqq R^\Free_\mK(\mE)$.
It is easily seen from $R^\Free_\mK(\mE) \ge 0$ that there exists $Z^\opt \in \Free$ such that
$(1 + \lambda^\opt)Z^\opt - \mE \in \mK$.
Let $p \coloneqq (1 - \delta)/(\lambda^\opt + 1 - \delta)$
and $Z' \coloneqq p(1 + \lambda^\opt)Z^\opt + (1-p)\delta Z$;
then, we have $0 \le p \le 1$, $Z' \in \Free$, and $Z' - \mE \in \mK$.
This implies $R^\Free_\mK(\mE) = 0$.]

We obtain the following proposition.
\begin{proposition} \label{pro:R_disc_ex}
 Let us consider $\mE \in \mN$.
 Let $\tV'$ be an arbitrary system.
 We consider a set of pairs $\mL \coloneqq \{(\{\hcJ_m\}_{m=1}^M, \{\Phi_m\}_{m=1}^M)\}$,
 where $\{ \hcJ_m : \Her_\tV \to \Her_{\tV'} \}_{m=1}^M$ is a collection of linear maps
 and $\Phi_1,\dots,\Phi_M \in \Her_{\tV'}$.
 Assume that the cone generated by
 \begin{alignat}{1}
  \mX &\coloneqq \left\{ \summ \hcJ_m^\dagger(\Phi_m) :
  ~(\{\hcJ_m\}, \{\Phi_m\}) \in \mL \right\}
 \end{alignat}
 is $\mK^*$, where $\hcJ_m^\dagger$ is the adjoint of $\hcJ_m$,
 which is defined as $\braket{\hcJ_m^\dagger(\Phi'),\mE'} = \braket{\Phi',\hcJ_m(\mE')}$
 $~(\forall \mE' \in \Her_\tV, \Phi' \in \Her_{\tV'})$.
 Then, we have
 \begin{alignat}{1}
  \max_{(\{\hcJ_m\}, \{\Phi_m\}) \in \mL'} \frac{\summ \braket{\Phi_m,\hcJ_m(\mE)}}
  {\max_{Z \in \Free} \summ \braket{\Phi_m,\hcJ_m(Z)}} &= 1 + R^\Free_\mK(\mE),
  \nonumber \\
  \label{eq:R_disc_ex}
 \end{alignat}
 where
 \begin{alignat}{1}
  \mL' &\coloneqq \left\{ (\{\hcJ_m\}, \{\Phi_m\}) \in \mL :
  \summ \hcJ_m^\dagger(\Phi_m) \neq \zero \right\}.
 \end{alignat}
\end{proposition}
\begin{proof}
 The left-hand side of Eq.~\eqref{eq:R_disc_ex} is rewritten by
 \begin{alignat}{1}
  \max_{\varphi \in \mX \setminus \{\zero\}}
  \frac{\braket{\varphi,\mE}}{\max_{Z \in \Free} \braket{\varphi,Z}}.
  \label{eq:R_disc_ex2}
 \end{alignat}
 Thus, an application of Lemma~\ref{lemma:R_disc} completes the proof.
\end{proof}

The operational meaning of Eq.~\eqref{eq:R_disc_ex} is as follows.
Suppose that $\{ \hcJ_m \}_m$ is a collection of (unnormalized) processes such that
$\summ \hcJ_m$ is a comb from $\Pos_\tV$ to $\Pos_{\tV'}$ and that $\{\Phi_k\}_k$ is a tester,
where the pair $(\{ \hcJ_m \}_m, \{\Phi_k\}_k)$ is restricted to belong to $\mL$.
We consider the situation that a party, Alice, applies a process $\hcJ_m$ to a comb
$\mE \in \Pos_\tV \cap \mN$,
and then another party, Bob, applies a tester $\{\Phi_k\}_k$ to $\hcJ_m(\mE)$.
The probability of Bob correctly guessing which of the processes $\hcJ_1,\dots,\hcJ_M$ Alice applies is
expressed by $\summ \braket{\Phi_m,\hcJ_m(\mE)}$
[note that $\sum_{k=1}^M \summ \braket{\Phi_k,\hcJ_m(\mE)} = 1$ holds].
Equation~\eqref{eq:R_disc_ex} implies that the advantage of $\mE$ over all $Z \in \Free$
in such a discrimination problem can be exactly quantified by the robustness $R^\Free_\mK(\mE)$.
In this situation, $\mK^* \subseteq \Pos^*_\tV$, i.e., $\mK \supseteq \Pos_\tV$, holds.

We give two examples of the application of Proposition~\ref{pro:R_disc_ex}.
The first example is the case $\mK = \Pos_\tV$.
Let us consider the case where $\summ \hcJ_m$
can be any comb from $\Pos_\tV$ to $\Pos_{\tV'}$
and $\{ \Phi_k \}_k$ can be any tester.
We can easily see that Eq.~\eqref{eq:R_disc_ex} with $\mK = \Pos_\tV$ holds.
Note that, for example, Theorem~2 of Ref.~\cite{Tak-Reg-Bu-Liu-2019} and
Theorems~1 and 2 of Ref.~\cite{Uol-Kra-Sha-Yu-2019} can be understood as
special cases of Proposition~\ref{pro:R_disc_ex} with $\mK = \Pos_\tV$.
The second example is the case $\mK \neq \Pos_\tV$.
For instance, for a given channel $\hmE$ from a system $\V$ to a system $\W$,
assume that $\hcJ_m$ is the process that applies $\hmE$ to a state
$\rho_m \in \Den_\V$ with probability $p_m$
[i.e., $\hcJ_m(\mE) = p_m \Trp{\V}[(I_\W \ot \rho_m^\T)\mE]$ holds]
and $\{ \Phi_k \}_k$ is a measurement of $\W$.
Then, we have
$\braket{\Phi_k, \hcJ_m(\Endash)} = p_m \braket{\Phi_k \ot \rho_m^\T, \Endash}$.
It is easily seen that Eq.~\eqref{eq:R_disc_ex} with $\mK^* = \Sep_{\W,\V}$
(or, equivalently, $\mK = \Sep_{\W,\V}^*$) holds,
where $\Sep_{\W,\V}$ is the set of all bipartite separable elements in $\Pos_{\W \ot \V}$.
Note that, for a linear map $\hPsi$ from $\V$ to $\W$,
$\Psi \in \Sep_{\W,\V}^*$ holds if and only if $\hPsi$ is a positive map.

%


\end{document}